\documentclass[journal,twoside,10pt]{IEEEtran}

\newtheorem{cor}{Corollary}
\newtheorem{ppro}{Proposition}

\IEEEoverridecommandlockouts

\usepackage{bm,cite,algorithm,algorithmic,float,amsmath,amssymb}
\usepackage[dvips]{graphicx}
\usepackage{subfigure,amsfonts}
\usepackage{cite, color}
\usepackage{balance}
\usepackage[left=0.65in, right=0.67in, top=0.9in, bottom=0.9in]{geometry}
\floatname{algorithm}{Algorithm}
\begin{document}

\title{Relay Control for Full-Duplex Relaying with Wireless Information and Energy Transfer}

\author{
\begin{center}
Hongwu~Liu,~\IEEEmembership{Member,~IEEE,}
Kyeong~Jin~Kim,~\IEEEmembership{Senior Member,~IEEE,}\\
H. Vincent Poor,~\IEEEmembership{Fellow,~IEEE,}
and~Kyung~Sup~Kwak,~\IEEEmembership{Member,~IEEE}
\end{center}
\thanks{H. Liu is with Shandong Jiaotong University, Jinan, China (e-mail: hong.w.liu@hotmail.com).}
\thanks{K. J. Kim is with Mitsubishi Electric Research Laboratories (MERL), Cambridge, MA, USA (e-mail: kyeong.j.kim@hotmail.com).}
\thanks{H. V. Poor is with  Department of Electrical Engineering, Princeton University, Princeton, NJ 08544.}
\thanks{K. S. Kwak is with Inha University, Incheon, Korea (e-mail: kskwak@inha.ac.kr).}
\thanks{This work was supported in part by SRF for ROCS, SEM, Shandong Provincial Natural Science Foundation, China (2014ZRB019XM), U. S. National Science Foundation under Grant ECCS-1343210, National Research Foundation of Korea-Grant funded by the Korean Government (Ministry of Science, ICT and Future Planning-NRF-2014K1A3A1A20034987).
}
}

\maketitle
\setcounter{page}{1}

\begin{abstract}
This study investigates wireless information and energy transfer for dual-hop amplify-and-forward full-duplex relaying systems. By forming energy efficiency (EE) maximization problem into a concave fractional program of transmission power, three relay control schemes are separately designed to enable energy harvesting and full-duplex information relaying.
With Rician fading modeled residual self-interference channel, analytical expressions of outage probability and ergodic capacity are presented for the maximum relay, signal-to-interference-plus-noise-ratio (SINR) relay, and target relay.
It has shown that EE maximization problem of the maximum relay is concave for time switching factor, so that bisection method has been applied to obtain the optimized value. By incorporating instantaneous channel information, the SINR relay with collateral time switching factor achieves an improved EE over the maximum relay in delay-limited and delay-tolerant transmissions. Without requiring channel information for the second-hop, the target relay ensures a competitive performance for outage probability, ergodic capacity, and EE. Comparing to the direct source-destination transmission, numerical results show that the proposed relaying scheme is beneficial in achieving a comparable EE for low-rate delay-limited transmission.
\end{abstract}


\section{Introduction}

Energy harvesting (EH) has emerged as a promising approach to prolong the lifetime of energy constrained wireless communications \cite{Energy_harvesting,Energy_harvesting_relay,Energy_harvesting_cognitive}. Through harvesting energy from natural sources (e.g., solar, wind, thermoelectric effects or other physical phenomena), periodic battery replacement or recharging can be alleviated in green communications \cite{SWIPT_magazine, Energy_harvesting,Energy_harvesting_relay, Energy_harvesting_cognitive, Energy_harvesting_circuit}. However, EH from natural sources is not typically applicable for devices with small sizes or located in physically limited places \cite{SWIPT_magazine, Energy_harvesting_circuit}.
With the capability to harvest energy from ambient radio-frequency (RF) signals, simultaneous wireless information and energy transfer (SWIET), also known as simultaneous wireless information and power transfer (SWIPT), provides a more promising way for green wireless communications to function in environments with physical or other limitations \cite{Varshney_SWIET,Grover_Shannon_meets_Tesla,SWIPT_architecture,SWIPT_cellular,SWIPT_magazine, Energy_harvesting_circuit, SWIPT_protocol_DF,SWIPT_protocol_AF}.

The pioneering work on SWIET can be traced back to \cite{Varshney_SWIET} and \cite{Grover_Shannon_meets_Tesla}, where the fundamental tradeoff between capacity and energy was studied for point-to-point communications. Following the assumption that an ideal receiver is capable of observing information and extracting energy from the same received signal, SWIET has been extended to multi-antenna systems \cite{R_Beam_WIPT,DZJPSSP}, multiuser systems \cite{SWIPT_multiuser}, and bi-directional communication systems \cite{I_J_TEI}. However, as discussed in \cite{SWIPT_architecture}, a practical circuit for EH from the RF signal can hardly decode the carried information from the same signal. Therefore, two practical receiver architectures, namely, time switching (TS) and power splitting (PS), are proposed in \cite{SWIPT_architecture}. They are now widely adopted in various wireless systems, such as multiple-input multiple-output (MIMO) systems \cite{MIMO_B_SWIPT}, orthogonal frequency division multiplexing systems \cite{SWIPT_OFDM}, and cellular systems \cite{SWIPT_cellular}\footnote{Note that antenna switching can also be applied in the multiple-antenna case \cite{DZJPSSP}.} .

In parallel with the aforementioned studies that mainly deal with single-hop scenarios, employment of cooperative relays to facilitate RF EH and information transfer in energy-efficient cooperative or green networks has also drawn significant attention \cite{SWIPT_protocol_DF,SWIPT_protocol_AF,SWIPT_magazine,SWIET_MIMO_OFDM_AF}. Relay-based SWIET not only enables wireless communications over long distances or across barriers, but also keeps the energy-constrained relays active through RF EH. The authors of \cite{SWIPT_protocol_AF} designed and analyzed the TS and PS relaying protocols for amplify-and-forward (AF) relaying systems, and then extended the results to an adaptive TS relaying protocol \cite{SWIPT_protocol_adptive}. 
The throughput of the TS and PS relaying protocols for decode-and-forward (DF) relaying was investigated in \cite{SWIPT_protocol_DF}. Several power allocation schemes for EH relay systems with multiple source-destination pairs were studied in \cite{SWIPT_PA}. The outage and diversity of SWIET for cooperative networks with spatially random relays were investigated in \cite{SWIPT_SRR} and the distributed PS-based SWIET was designed for   interference relay systems \cite{SWIPT_DPS}. More recently, antenna switching and antenna selection for the SWIET relaying systems have been investigated in \cite{SWIPT_antenna_switch} and \cite{SWIPT_Antenna_select}, respectively. Nevertheless, all these studies are limited to half-duplex relaying (HDR) mode. Since the source-to-relay and relay-to-destination channels are kept orthogonal by either frequency division or time division multiplexing, significant loss of spectral efficiency occurs in the HDR mode. As an alternative, full-duplex relaying (FDR) has drawn considerable attention \cite{Inband_FD, Hybird_FD_HD, SWIPT_FD,Resource_FDR, Cognitive_FDR,energy_harvest:poor}. Since FDR requires only one channel for the end-to-end transmission, a significant improvement in the spectral efficiency over the HDR can be achieved.

So far, a few studies have been conducted for SWIET in FDR systems. In \cite{SWIPT_FD}, the throughput of the TS relaying protocol has been analyzed for SWIET FDR systems, in which the EH relay is operated cooperatively. In practice, since the relay node suffers severe self-interference from its own transmit signal, FDR operation is difficult to implement. For example, self-interference of more than 106 dB  has to be suppressed by a femto-cell FDR base-station to achieve the link signal-to-noise-ratio (SNR) equals to that of an HDR
counterpart \cite{Femto_SI}. For systems that require a higher transmission power, more self-interference suppression is needed \cite{Inband_FD}.
Therefore, MIMO has been employed at EH FDR node to suppress self-interference\cite{SWIPT_FDR_MIMO}.
By employing EH relay in the second time phase for the conventional two-phase AF HDR systems, a self-interference immunizing full-duplex relay node was proposed in \cite{SWIPT_FD_selfenergy}, which can transmit information and extract energy simultaneously via separated transmit and receive antennas. Note that all the above researches of EH FDR are conducted to maximize throughput, whereas energy efficiency (EE) has not been investigated.
Another challenging problem for SWIET is to determine the EH parameters, i.e., the TS factor and PS factor for the TS and PS relaying protocols, respectively. Determining of the TS factor affects not only the relay-harvested energy but also the effective relaying transmission time in the TS relaying protocol. Compared with the PS relaying protocol, the TS relaying protocol is more practical because of its simplicity.
With statistical channel state information (CSI), the numerical optimizations of the TS factor  have been presented in \cite{SWIPT_protocol_AF} and \cite{SWIPT_FD}.
In delay-limited and delay-tolerant transmissions, the instantaneous CSI can also be applied to optimize the EH parameter, as well as the instantaneous CSI-aided transmission power control in the FDR systems \cite{Cognitive_FDR, Resource_FDR}.

Motivated by these previous studies, we focus on wireless information and energy transfer for an energy-constrained dual-hop FDR system, in which the source node has a reliable green power supply, whereas the TS-based AF relay node is powered via EH from the source-emitted RF signal\footnote{This setting has a number of potential applications in green wireless network, e.g., when the intermediate node is energy-selfish or lacks an energy-supply, or the direct link from the source to destination is blocked by a barrier while the relay has to be placed on a site without a fixed power supply.}.
Since self-interference deteriorates FDR performance seriously, relay gain control with SWIET becomes more critical than ever before. Compared with existing works, some distinct features of our study are highlighted here.
In \cite{SWIPT_FD_selfenergy}, the effective information transmission time is the same as that of HDR systems. Therefore, the spectral efficiency improvement is minimal compared with that of the FDR systems. In our study, we focus on improving EE by applying full-duplex information relaying, i.e., the relay receives and forwards the source information to the destination simultaneously.
In \cite{SWIPT_FD}, the authors designed wireless information and energy transfer for FDR systems with an aim to maximize throughput. For full-duplex radios, experiments have shown that residual self-interference (RSI) contains specular component and RSI channel is characterized by Rician fading with its $K$-factor depending on the amount of direct-path suppression \cite{Full_duplex_experiment}. Similar trends for RSI characters can also be observed in \cite{Full_duplex_radios}. Since EE is a major concern in the design of SWIET FDR, we focus on the impact of this experimental verified RSI channel model on EE performance by optimizing TS factor, relay gain, and source transmission power.

In this paper, three relay control schemes, namely, the maximum relay, signal-to-interference-plus-noise-ratio (SINR) relay, and target relay, are investigated for the three transmission schemes of instantaneous transmission, delay-limited transmission, and delay-tolerant transmission. The contributions of this study are summarized as follows:
\begin{itemize}
\item Considering that FDR RSI channel follows Rician fading, we present analytical expressions for the throughput for the three relay control schemes in different transmissions, such as instantaneous transmission, delay-limited transmission, and delay-tolerant transmission. Specifically, analytical expressions for the outage probability are derived for these three relay control schemes in delay-limited transmission, while analytical expressions for the ergodic capacity are derived for the maximum relay and SINR relay in delay-tolerant transmission.
\item We show that EE maximization with given TS factor and relay gain is a concave fractional program of transmission power, so that EE monotonically decreases with an increasing transmission power in instantaneous and delay-tolerant transmissions for practical transmission power\footnote{Transmission power is not very small, e.g., not less than 0 dBm.}. In delay-limited transmission, we show that a relatively higher outage probability, as well as correspondingly a relatively lower transmission power, may achieve a higher EE. When the source-destination direct link is not available, the numerical results also show that SWIET FDR is beneficial in achieving EE for low-rate delay-limited transmission. By enhancing the performances of outage probability, ergodic capacity, effective relaying time, and corresponding EE, the SINR relay is shown to be more competitive than the maximum relay. Without requiring CSI of the second-hop, the target relay also achieves a competitive outage probability, ergodic capacity, and EE. We also reveal that the worst performances of outage probability, ergodic capacity, and EE are obtained by the relay placed midway between the source and destination.
\item It has shown that the EE maximization problem of the maximum relay is concave for all possible values of the TS factor. Thus, the low-complexity bisection method can be applied to obtain the optimized TS factor. The closed-form TS factors have been designed for the SINR relay and target relay to maximize the e-SINR and achieve a target e-SINR, respectively. With the obtained closed-form TS factors for the SINR relay and target relay, instantaneous CSI has been employed to improve the EE in delay-limited and delay-tolerant transmissions.
\end{itemize}

The rest of this paper is organized as follows. Section II describes the system model of the considered FDR system and formulates the EE optimization problem. Section III presents the three relay control schemes. The analytical results of the throughput are presented in Section IV. Section V presents numerical results and discusses the system performances of our proposed relay control schemes. Finally, Section VI summarizes the contributions of our study.

\section{System Model}

In this paper, we consider a wireless dual-hop FDR system, in which a source node intends to transfer its information to the destination node. Due to physical isolations or environmental limitations between the source and destination, a cooperative relay is employed to assist information transmission from the source to the destination. The cooperative relay is assumed to be an energy-selfish or energy-constrained device such that it needs to harvest energy from the source-emitted RF signal to forward the source information to the destination. For simplicity of implementation, the AF relaying scheme and TS transceiver architecture are chosen at the relay node. The channel coefficients from the source to the relay and from the relay to the destination are denoted by $h_1$ and $h_2$, respectively. $h_1$ and $h_2$ are assumed to be frequency non-selective and quasi-static block-fading, following a Rayleigh distribution. The means of the exponential random variables $|h_1|^2$ and $|h_2|^2$ are denoted by $\lambda _1$ and $\lambda _2$, respectively. According to the experimental results of \cite{Full_duplex_experiment}, the RSI at a full-duplex node contain specular components such that the RSI channel coefficient can be modeled as a Rician random variable $h_0$. Since both experiment and analysis have shown that RSI cannot be eliminated completely, this paper considers only the case of $|h_0|^2>0$ \cite{Full_duplex_experiment, RSI_after_cancellation}.

\begin{figure}[htbp]
\vspace{0.12in}
\begin{center}
\includegraphics[width=2in]{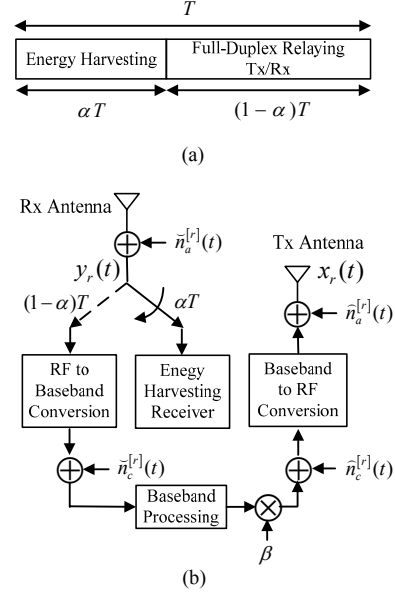}
\vspace{-0.09in}
\caption{(a) Illustration of the key parameters for EH and FDR at the relay. (b) Block diagram of the relay transceiver.}
\end{center}
\vspace{-0.15in}
\label{fig:1}
\end{figure}

The framework of the TS relaying protocol is illustrated in Fig. 1(a), in which each time block $T$ is divided into two phases.
Denoting the TS factor by $\alpha$ ($0<\alpha<1$), we use the first phase assigned with a duration of $\alpha T$ for energy transfer from the source to the relay. The second phase assigned with the remaining duration of $(1-\alpha)T$ is used for full-duplex information relaying via the dual-hop channel.
The relay-received RF signal in the two time phases are sent to the EH receiver and full-duplex transceiver, respectively, as illustrated in Fig. 1(b).
Since the relay node does not transmit during the first time phase, self-interference is not introduced during the EH period. The harvested energy at the relay is given by
\begin{eqnarray}
{E} = {\eta \mathcal{L}_1 {d_1^{-m}} P_s  |h_1|^2} \alpha T,
\end{eqnarray}
where $\eta$ is the energy conversion efficiency depending on the rectifier circuit, $\mathcal{L}_1 $ is a path loss factor, $d_1$ is the distance between the source and relay and it has been normalized with respect to the reference distance, $m$ is the path loss exponent, and $P_s$ is the source transmission power. In this paper, the path loss factor is defined as $\mathcal{L}_1 \triangleq \mathcal{A}_s \mathcal{L}$, where $\mathcal{A}_s$ is the source transmit antenna gain and $\mathcal{L}$ is the measured path loss at the reference distance.
Given the relay-harvested energy, the maximum transmission power at the relay is expressed as
\begin{eqnarray}
P_{_{E}} = \tfrac{{{E}}}{{(1 - \alpha )T}} = {{\mu \mathcal{L}_1 d_1^{-m} {P_s}|h_1|^2}}, \label{eq:P_max}
\end{eqnarray}
where $\mu  \triangleq \frac{{\alpha \eta }}{{1 - \alpha }}$.

In the FDR mode, the relay concurrently receives the signal $y_r(t)$ and transmits the signal $x_r(t)$ on the same frequency.
As depicted in Fig. 1(b), the full-duplex transceiver down-converts the received RF signal to the baseband, processes the baseband signal, and up-converts the processed baseband signal. In Fig. 1(b),  $\mathord{\buildrel{\lower3pt\hbox{$\scriptscriptstyle\smile$}}
\over n} _a^{_{[ r ]}}(t)$ and $\mathord{\buildrel{\lower3pt\hbox{$\scriptscriptstyle\frown$}}
\over n} _a^{_{[ r ]}}(t)$ are the narrow-band Gaussian noises introduced by the receive and transmit antennas, respectively.
In addition,
$\mathord{\buildrel{\lower3pt\hbox{$\scriptscriptstyle\smile$}}
\over n} _c^{_{[ r ]}}(t)$ and $\mathord{\buildrel{\lower3pt\hbox{$\scriptscriptstyle\frown$}}
\over n} _c^{_{[ r ]}}(t)$ are the baseband additive white Gaussian noises (AWGNs) caused by down-conversion and up-conversion, respectively \cite{SWIPT_architecture}. For simplicity, the equivalent baseband noise composing both $\mathord{\buildrel{\lower3pt\hbox{$\scriptscriptstyle\smile$}}
\over n} _a^{_{[ r ]}}(t)$ and $\mathord{\buildrel{\lower3pt\hbox{$\scriptscriptstyle\frown$}}
\over n} _a^{_{[ r ]}}(t)$ is modeled by the zero mean AWGN $n_a^{_{[ r ]}}(t)$ with the variance $\sigma _a^2$, and the equivalent baseband noise composing both
$\mathord{\buildrel{\lower3pt\hbox{$\scriptscriptstyle\smile$}}
\over n} _c^{_{[ r ]}}(t)$ and $\mathord{\buildrel{\lower3pt\hbox{$\scriptscriptstyle\frown$}}
\over n} _c^{_{[ r ]}}(t)$ is modeled by the zero mean AWGN $n_c^{_{[ r ]}}(t)$ with the variance $\sigma _c^2$. Therefore, the overall AWGN at the relay node can be modeled as the zero mean AWGN $n_r(k) \triangleq n_a^{_{[ r ]}}(k) + n_c^{_{[ r ]}}(k)$ with the variance $\sigma _r^2 \triangleq \sigma _a^2 + \sigma _c^2$.
At the relay node, the sampled baseband signal is given by
\begin{eqnarray}
{y_r}(k) = \sqrt {{\tfrac{\mathcal{L}_1 {P_s}}{d_1^m}}} h_1 s(k) + \mathcal{A}_r h_0 x_r(k) + {n_r}(k), \label{eq:y_r}
\end{eqnarray}
where $k$ denotes the symbol index, $s(k)$ is the sampled and normalized information signal from the source, $x_r(k)$ is the sampled signal of $x_r(t)$, and $\mathcal{A}_r$ is the relay transmit antenna gain. Note that the second term on the right side of \eqref{eq:y_r} is the RSI.

Using the harvested energy, the relay amplifies the received signal by a relay gain $\beta$. Then, the transmitted signal at the relay can be expressed as
\begin{eqnarray}
{x_r}(k) = \sqrt{\beta} {y_r}(k- \tau), \label{eq:x_r}
\end{eqnarray}
where $\tau \ge 1$ is the processing delay at the relay.
By recursively substituting \eqref{eq:y_r} and \eqref{eq:x_r}, we have the following expression for the  transmitted signal at the relay:
\begin{eqnarray}
{x_r}(k) &\!\!=\!\!& \sqrt \beta  \sum\limits_{l = 1}^\infty  {{(\mathcal{A}_r h_0 \sqrt \beta  )}^{l - 1}} \nonumber  \\
 &\!\!\!\!& \times \left( {\sqrt {\tfrac{\mathcal{L}_1   P_s}{{d_1^m}}} h_1s(k - l\tau ) + {n_r}(k - l\tau )} \right).  \label{eq:x_r_recursive}
\end{eqnarray}
The sampled received signal at the destination is given by
\begin{eqnarray}
{y_d}(k) = \sqrt{\tfrac{\mathcal{L}_2}{{d_2^m}}}h_2{x_r}(k) + n_d(k), \label{eq:y_d}
\end{eqnarray}
where $d_2$ is the distance between the relay and  destination, which is also normalized with respect to the reference distance,  $\mathcal{L}_2 \triangleq \mathcal{A}_r \mathcal{L}$ is a path loss factor,  and $n_d(k)$ is the AWGN with the zero mean and the variance $\sigma _d^2$. By substituting \eqref{eq:x_r_recursive} into \eqref{eq:y_d}, we have
\begin{eqnarray}
\!\!\!\!{y_d}(k) & \!\!\!\!\!=\!\!\!\!\!& \sqrt {\tfrac{{{\mathcal{L}_1 \mathcal{L}_2 P_s}\beta }}{{d_1^md_2^m}}}  h_1 h_2 \sum\limits_{l = 1}^\infty  {{{(\mathcal{A}_r h_0\sqrt \beta  )}^{l - 1}}s(k - l\tau )}  \nonumber \\
\!\!\!\!&\!\!\!\!\!\!\!\!\!\! & + \sqrt {\tfrac{\mathcal{L}_2   \beta }{{d_2^m}}} h_2 \!\! \sum\limits_{l = 1}^\infty  {{{(\mathcal{A}_r h_0\sqrt \beta  )}^{l - 1}}{n_r}(k - l\tau )}  + {n_d}(k). \label{eq:y_d_recursive}
\end{eqnarray}

In the following, we derive the end-to-end signal power under the condition of employing cooperative non-oscillatory relays. By assuming that all the signal and noise samples are mutually independent, we calculate the relay transmission power $P_r = \mathbb{E} \{ |{x_r}(k){|^2}\}$ from \eqref{eq:x_r_recursive} as
\begin{eqnarray}
P_r  &=& {\beta}\sum\limits_{l = 1}^\infty  {{{(\mathcal{A}_r |h_0{|^2}{\beta})}^{l - 1}}\left( { \mathcal{L}_1  d_1^{-m}{P_s}|h_1{|^2} + \sigma _r^2} \right)} \nonumber  \\
 &=& {\beta}\frac{{\mathcal{L}_1 d_1^{-m}{P_s}|h_1{|^2} + \sigma _r^2}}{{1 - \mathcal{A}_r |h_0{|^2}{\beta}}}.   \label{eq:P_r}
\end{eqnarray}
To prevent oscillation and guarantee a finite relay transmission power, the relay gain is limited by
\begin{eqnarray}
{\beta} < \frac{1}{{\mathcal{A}_r|h_0{|^2}}}.  \label{eq:beta2<1/f2}
\end{eqnarray}
The actual relay transmission power should be less than or equal to the maximum transmission power, i.e.,
\begin{eqnarray}
P_{r} \le P_{_{E}}. \label{eq:P_r<P_max}
\end{eqnarray}
When \eqref{eq:P_max} and \eqref{eq:P_r} are substituted into \eqref{eq:P_r<P_max}, the relay gain under the maximum relay transmission power is limited by
\begin{eqnarray}
\beta  \le \frac{\mu }{{1 + \gamma _{{\rm{SR}}}^{ - 1} + \mu \mathcal{A}_r |h_0|^2}}, \label{eq:beta2<beta_max}
\end{eqnarray}
where the channel SNR of the first-hop link is defined as  $\gamma _{_{\rm SR}} \triangleq {\mathcal{L}_1 {P_s |h_1|^2}\over {d_1^m \sigma _r^2}}$.
At symbol index $k$, the destination node can employ any standard detection procedure to decode the desired signal $s(k-\tau)$, and the rest of the received signal components act as interference and noise. Again based on the assumption that signal and noise are independent of each other, the received signal power at the destination is calculated from \eqref{eq:y_d} as ${\mathbb{E}} \{|y_d(k)|^2 \}= \mathcal{L}_2 d_2^{-m} |h_2|^2 {\mathbb{E}}\{|x_r(k)|^2\}+ \sigma _d^2$, which can be further evaluated as follows, comprising of the desired signal power, RSI power, and noise power:
\begin{eqnarray}
{\mathbb{E}} \{|y_d(k)|^2 \}
 &\!\!\!\!=\!\!\!\!& \underbrace {\mathcal{L}_1 \mathcal{L}_2 {d_1^{-m} d_2^{-m}}{\beta}{P_s}|h_1|^2|h_2|^2}_{{\rm{desired\;signal\;power}}} \nonumber \\
 &\!\!\!\!\!\!\!\!&\!\!\!\!\!\!\!\!\!\!\!\!\!\!\!\!+ \underbrace {\left( {{\mathcal{L}_1 {d_1^{-m}}}{P_s}|h_1|^2 + \sigma _r^2} \right){\mathcal{L}_2 }{{d_2^{-m}}}\beta|h_2|^2\tfrac{{\mathcal{A}_r |h_0|^2{\beta}}}{{1 - \mathcal{A}_r|h_0|^2{\beta}}}}_{{\rm{RSI\; power}}} \nonumber \\
 &\!\!\!\!\!\!\!\!&\!\!\!\!\!\!\!\!\!\!\!\!\!\!\!\! + \underbrace {{\mathcal{L}_2 }{{d_2^{-m}}}\beta|h_2|^2\sigma _r^2 + \sigma _d^2}_{{\rm{noise\; power}}}.   \label{eq:y_d_power}
\end{eqnarray}
Based on \eqref{eq:y_d_power}, the e-SINR at the destination is given by
\begin{eqnarray}
\gamma _{_{\rm SRD}}  = \frac{{{\gamma _{_{\rm{SR}}}}{\gamma _{_{\rm{RD}}}}}}{{{\gamma _{_{\rm{SR}}}}/\beta  + {\gamma _{_{\rm{RD}}}} + ({\gamma _{_{\rm{SR}}}} + 1){\gamma _{_{\rm{RD}}}}\frac{{\mathcal{A}_r|h_0|^2}}{{1/\beta  - \mathcal{A}_r|h_0|^2}}}},  \label{eq:SINR}
\end{eqnarray}
where the channel SNR of the second-hop link is defined as ${\gamma _{_{\rm{RD}}}} \triangleq \frac{{{P_{_E}}}}{\mu }\frac{\mathcal{L}_2 {|h_2|^2}}{{d_2^m\sigma _d^2}} = \frac{{\mathcal{L}_1 \mathcal{L}_2 {P_s}|{h_1}|^2|h_2|^2}}{{d_1^md_2^m\sigma _d^2}}$.

In this paper, the EEs of instantaneous transmission, delay-limited transmission, and delay-tolerant transmission are considered. The EE is defined as the number of bits successfully conveyed to the destination per Joule consumed energy and is given by
\begin{eqnarray}
\eta _{_{\rm eff}}(P_s, \alpha, \beta) = {B R_{_{\rm M}}(P_s, \alpha, \beta)}/{P_s},  \label{eq:energy_efficiency}
\end{eqnarray}
where $B$ is the system bandwidth and $R_{_{\rm M}}(P_s, \alpha ,\beta )$ is the throughput of the transmission scheme $\rm M$. Let $\rm{M} = {\rm I}, {\rm DL}$, and ${\rm DT}$ represent instantaneous transmission, delay-limited transmission, and delay-tolerant transmission, respectively. Their corresponding throughputs are respectively given by
\begin{subequations}
\begin{eqnarray}
{R_{_{\rm{I}}}}(P_s, \alpha ,{\beta}) &=& (1 - {\alpha }){\log _2}(1 + {\gamma }), \label{eq:R_I}\\
{R_{_{{\rm{DL}}}}}(P_s, \alpha ,{\beta}) &=& (1 - {\alpha })(1 - {P_{{\rm{out}}}})R, \label{eq:R_DL}\\
{R_{_{\rm{DT}}}}(P_s, \alpha ,{\beta}) &=& (1 - \alpha ){C_{_{\rm{E}}}}, \label{eq:R_DT}
\end{eqnarray}
\end{subequations}
where $P_{\rm out} =\Pr(\gamma < \gamma _{\rm th})$ is the outage probability, $R = \log _2(1+\gamma _{\rm th})$ is the fixed transmit rate, $C_{_{\rm E}} = {\mathbb{E}} \{ \log _2(1+\gamma )  \}$ is the ergodic capacity, and $\gamma _{\rm th}$ is an e-SINR threshold for correct data detection at the destination. The design goal is to maximize the EE by optimizing the control parameters $\{P_s, \alpha, \beta\}$.
The optimal control parameters $\{P_s ^*,  \alpha ^*, \beta ^*\}$ can be obtained by solving the following optimization problem:
\begin{subequations}
\begin{eqnarray}
\!\!\!\!& &\!\!\!\!\{P_s^*,  {\alpha ^*},{\beta ^*}\}  = \arg \mathop {\max }\limits_{\alpha ,\beta } {\eta_{_{\rm eff}}}(P_s, \alpha ,\beta ) \label{eq:P1}\\
\!\!\!\!& &{\rm subject~to}~~~\!0<P_s<P_{\max},  \\
\!\!\!\!& &~~~~~~~~~~~~~~~ 0<\alpha<1, \label{eq:a_constaint} \\
\!\!\!\!& &~~~~~~~~~~~~~~~ 0< \beta \le \tfrac{\mu }{{1 + \gamma _{_{\rm{SR}}}^{ - 1} + \mu |h_0|^2}}. \label{eq:b_constaint}  \end{eqnarray}
\end{subequations}
Since ${\eta_{_{\rm eff}}}$ has a very complicated express with respect to $\{P_s, \alpha, \beta\}$ and the information of $\{\alpha, \beta\}$ is unknown, the optimal control parameters $\{P_s ^*,  \alpha ^*, \beta ^*\}$ can hardly get a closed-form solution. By verifying the second derivative of $\gamma _{_{\rm SRD}}$ with respect to $P_s$, it can be shown $\gamma _{_{\rm SRD}}(P_s)$ is a concave function when $\beta$ satisfies the non-oscillation condition \eqref{eq:beta2<1/f2}. Since $\log_2(1+\gamma_{_{\rm SRD}})$ and $1-\Pr (\gamma_{_{\rm SRD}} < \gamma _{\rm th})$ in (15) are monotonically increasing with $\gamma _{_{\rm SRD}}$, $\log_2(1+\gamma_{_{\rm SRD}})$ and $1-\Pr (\gamma_{\max} < \gamma _{_{\rm SRD}})$ are also concave for all possible values of $P_s$. Thus, the numerator of ${\eta_{_{\rm eff}}} (P_s)$ is a concave function for any $\{ {\alpha}, \beta \}$ satisfying \eqref{eq:a_constaint} and \eqref{eq:b_constaint}.
Note that the denominator of ${\eta_{_{\rm eff}}} (P_s)$, $P_s$, is convex respect to $P_s$. Then, for any $\{ {\alpha}, \beta \}$ satisfying \eqref{eq:a_constaint} and \eqref{eq:b_constaint}, the optimization problem \eqref{eq:P1} becomes a concave fractional programing of $P_s$, which can be solved exactly by Dinkelbach's method \cite{Non_linear_FP}. The optimization problem can be reduced to optimizing $\{\alpha, \beta\}$, so that the system design becomes mathematically tractable.
Since $P_{\rm out}$ and $C_{\rm E}$ generally require only statistics of CSI, the appearance of the term $(1-\alpha)$ in \eqref{eq:R_DL} and \eqref{eq:R_DT} intrinsically  implies that $\alpha$ is optimized statistically for the delay-limited and delay-tolerant transmissions. However, as we will show lately, instantaneous CSI can also be employed to optimize $\alpha$ for delay-limited and delay-tolerant transmissions.
In the following, all the three relay control schemes and their EE-related performances are presented for the considered transmission schemes, respectively.

\section{Relay Control Scheme}

In this section, we investigate how to compute the relay gain and  TS factor for the three relay control schemes, namely, maximum relay, SINR relay, and target relay. In designing the relay control scheme, we assume that perfect knowledge of CSI is available.

\subsection{Maximum Relay}
A simple and popular relay control scheme is involved in setting the relay gain at the maximum relay transmission power \cite{Hybird_FD_HD, Gain_control_FDR, SWIPT_FD}. In contrast to these works, our study considers the maximum relay targeting at maximizing the EE in the presence of RSI. For a given $\alpha$ in the range $(0, 1)$, the relay-harvested energy and maximum relay transmission power are determined, so that
the relay gain is given by the following according to \eqref{eq:beta2<beta_max}:
\begin{eqnarray}
\beta _{\max } = \tfrac{{\mu {\gamma _{_{\rm{SR}}}}}}{{1+ {\gamma _{_{\rm{SR}}}} + \mu {\gamma _{_{\rm{SR}}}} \mathcal{A}_r|h_0|^2  }},  \label{eq:beta_max}
\end{eqnarray}
which guarantees that \eqref{eq:beta2<1/f2} holds. Substituting \eqref{eq:beta_max} into \eqref{eq:SINR}, the e-SINR achieved by the maximum relay is given by
\begin{eqnarray}
{\gamma _{\max }} = \tfrac{{\mu {\gamma _{_{\rm{SR}}}}{{\gamma }_{_{\rm{RD}}}}}}{{{\gamma _{_{\rm{SR}}}} + \left( {\mu {\gamma _{_{\rm{SR}}}}\mathcal{A}_r|h_0|^2 + 1} \right)\left( {\mu {{\gamma }_{_{\rm{RD}}}} + 1} \right)}}. \label{eq:SINR_max}
\end{eqnarray}
Now, the instantaneous throughput can be evaluated by substituting \eqref{eq:SINR_max} into \eqref{eq:R_I}.

The appearance of the term $\mathcal{A}_r |h_0|^2$ in \eqref{eq:beta_max} and \eqref{eq:SINR_max} indicates that the system performance of the maximum relay is affected by RSI. Since $h_0$ is a Rician variable, $|h_0|^2$ follows the non-central chi-squared distribution and its probability density function (PDF) is given by
\begin{eqnarray}
{f}_{_{|h_0|^2}}(w)\triangleq   \tfrac{(K+1)e^{-K}}{\sigma_0^2} e^{-\frac{(K+1)w}{\sigma_0^2}} I_0\left( 2 \sqrt{\tfrac{K(K+1)w}{\sigma_0^2}} \right),  \label{eq:pdf_f}
\end{eqnarray}
where $\sigma_0^2 = \mathbb{E} \{ |h_0|^2 \}$ is the average RSI channel gain and $K$ is the Rician $K$-factor.

\begin{ppro}
The outage probability achieved by the maximum relay is given by
\begin{subequations}
\begin{eqnarray}
\!\!\!\!\!\!\!\!\!\!\!\!\!\!{P_{{\rm{out}}}} &\!\!\!\!=\!\!\!\!&  1 \!-\! \frac{1}{\lambda _1}\int\limits_{w=0}^{\frac{{1}}{{\mu {\gamma _{{\rm{th}}}}}}} {\int\limits_{z = \frac{d}{c}}^\infty \!\! {{ {f}_{|h_0|^2}(w) e^{ - \left( {\frac{z}{{{\lambda _1}}} + \frac{{az + b}}{{(c{z^2} - dz){\lambda _2}}}} \right)}} {\rm{d}}z} } {\rm{d}}w  \label{eq:proposition1} \\
\!\!\!\!\!\!\!\!\!\!\!\!\!\!&\!\!\!\! \mathop  \approx \limits^{(a)}  \!\!\!\!&1 -  \int\limits_{0}^{{{{1 } \over {\mu {\gamma _{{\rm{th}}}}}}}} { {{{f}_{|h_0|^2}}} (w) \rho {K_1}(\rho ){e^{ - \frac{d}{{c{\lambda _1}}} }}} {\rm{d}}w,
\end{eqnarray}
\end{subequations}
where $a \triangleq \mathcal{L}_1  {P_s}d_1^md_2^m\sigma _d^2{\gamma _{\rm th}}(1 + \mu \mathcal{A}_r w)$, $b \triangleq d_1^{2m}d_2^m\sigma _r^2\sigma _d^2{\gamma _{\rm th}}$, $c \triangleq \mathcal{L}_1^2 \mathcal{L}_2 P_s^2\mu (1 - \mu {\gamma _{\rm th}}\mathcal{A}_r w)$, $d \triangleq \mathcal{L}_1 \mathcal{L}_2 {P_s}d_1^m\sigma _r^2\mu {\gamma _{\rm th}}$, $\rho \triangleq \sqrt{\frac{4a}{c\lambda _1 \lambda _2}} $, $K_1(\cdot)$ is the first-order modified Bessel function of the second kind \cite[Eq. (8.432)]{Table_Integrals}, and the approximation in the step (a) is achieved in the region of high SNR values.
\end{ppro}
\begin{proof}
A proof is provided in Appendix A.
\end{proof}

\begin{ppro}
The ergodic capacity achieved by the maximum relay is given by
\begin{eqnarray}
C_{_{\rm E}} &\!\!\!\!\!\!=\!\!\!\!\!\!& \tfrac{1}{{\ln 2}} \int\limits_0^\infty f_{_{|h_0|^2}}(w) G_{4,2}^{1,4}\left( {\left. {\tfrac{{\mathcal{L}_1 \mathcal{L}_2 \mu {P_s}{\lambda _1}{\lambda _2} (1+\mu \mathcal{A}_r w)}}{{d_1^md_2^m\sigma _d^2}}} \right|{}_{1,0}^{0,0,1,1}} \right) {\rm d}w \nonumber \\
&\!\!\!\!\!\!\!\! & - \tfrac{e^{-K}}{K \ln{2}}\sum\limits_{n = 0}^\infty  {\tfrac{{{{\left( {\frac{{d_1^md_2^m\sigma _d^2K(K + 1)}}{{{\mathcal{L}_1 \mathcal{L}_2} {\mu ^2}{P_s}{\lambda _1}{\lambda _2} \mathcal{A}_r\sigma _0^2}}} \right)}^{n + 1}}}}{{{{(n!)}^2}}}} \nonumber \\
&\!\!\!\!\! \!\!\!\!\!& \times
 G_{1,4}^{4,1}\left( {\left. {{\tfrac{d_1^md_2^m\sigma _d^2(K + 1)}  {{\mathcal{L}_1 \mathcal{L}_2} {\mu ^2}{P_s}{\lambda _1}{\lambda _2} \mathcal{A}_r \sigma _0^2}}} \right|{}_{0, - 1 - n, - 1 - n, - n}^{ - 1 - n}} \right),
\end{eqnarray}
where $G_{m,n}^{p,q}(x)$ is the Meijer G-function \cite[Eq. (9.301)]{Table_Integrals}.
\end{ppro}
\begin{proof}
A proof is provided in Appendix B.
\end{proof}

When the maximum relay is applied, the optimized TS factor can be obtained by solving the following optimization problem:
\begin{eqnarray}
& &{\alpha ^*} = \arg \mathop {\max }\limits_{\alpha } \eta _{_{\rm eff}}(\alpha)  \nonumber \\
& & {\rm{subject~to~}} 0 < \alpha  < 1. \label{eq:opt_prob}
\end{eqnarray}

By verifying the sign of $\frac{\partial \gamma_{\max}}{\partial \alpha }$, we have $\tfrac{{\partial {\gamma _{\max }}}}{{\partial \alpha }} > 0$ when
\[
0<\alpha <\tfrac{{\sqrt {{\gamma _{_{\rm{SR}}}} + 1} }}{{\sqrt {{\gamma _{_{\rm{SR}}}} + 1}  + \eta \sqrt {{\gamma _{_{\rm{SR}}}}{\gamma _{_{\rm{RD}}}}|\mathcal{A}_r h_0|^2} }}
\]
and $\tfrac{{\partial {\gamma _{\max }}}}{{\partial \alpha }} < 0$ when
\[
\tfrac{{\sqrt {{\gamma _{_{\rm{SR}}}} + 1} }}{{\sqrt {{\gamma _{_{\rm{SR}}}} + 1}  + \eta \sqrt {{\gamma _{_{\rm{SR}}}}{\gamma _{_{\rm{RD}}}}\mathcal{A}_r|h_0|^2} }}<\alpha<1.
\]
Thus, $\gamma _{\max}$ is concave for all possible values of $\alpha$. Since $\log_2(1+\gamma_{\max})$ and $1-\Pr (\gamma_{\max} < \gamma _{\rm th})$ in (15) are monotonically increasing with $\gamma _{\max}$, $\log_2(1+\gamma_{\max})$ and $1-\Pr (\gamma_{\max} < \gamma _{\rm th})$ are also concave for all possible values of $\alpha$. Based on the fact that the coefficient $(1-\alpha)$ in (15) does not change the convexity of $R_{_{\rm M}}$ inherited from $\log_2(1+\gamma_{\max})$ and $1-\Pr (\gamma_{\max} < \gamma _{\rm th})$, it can be concluded that $\eta_{_{\rm eff}}(\alpha)$ is concave. Then, the optimized $\alpha ^*$ can be obtained by solving the equation $\frac{\partial \eta_{_{\rm eff}}(\alpha)}{{ \partial \alpha }} = 0$. Given the complicated expression for $\frac{\partial \eta_{_{\rm eff}}(\alpha)}{{ \partial \alpha }} = 0$, the closed-form solution is difficult to obtain. However, since a unique global maximum of $\eta_{_{\rm eff}}(\alpha)$ exists, the optimized TS factor can be obtained by applying the bisection method with the complexity of ${\cal O}(\log \frac{1}{\varepsilon} )$, where $\varepsilon$ is the searching precision \cite{Boyd:Vandenberghe}. Since $P_{\rm out}$ and $C_{_{\rm E}}$ involve only statistics of CSI, the maximum relay takes no advantage of instantaneous CSI in the delay-limited and delay-tolerant transmissions, respectively.

\subsection{SINR Relay}

Since the EE is affected by both the e-SINR and relaying transmission time when the source transmission power is fixed, TS factor can be optimized to maximize the e-SINR and relaying transmission time jointly or separately. In this part, a relay control scheme is proposed to maximize the e-SINR with the aid of instantaneous CSI. Compared with the maximum relay which searches TS factor to optimize the e-SINR and relaying transmission time simultaneously, the proposed scheme only maximizes the e-SINR and we call it the SINR relay throughout the paper.

According to \eqref{eq:y_d_power}, given the received signal power at the destination node as a function of the relay gain, the desired signal power is linear, but the RSI power is nonlinear. Consequently, increasing the relay gain can increase the RSI power faster than a desired signal power and lead to a reduced e-SINR. We can show that \eqref{eq:SINR} has a single maximum point for $\beta \in (0, {1\over \mathcal{A}_r |h_0|^2})$. Thus, by setting the derivative of \eqref{eq:SINR} equals to zero, the relay gain of the SINR relay that maximizes the e-SINR is obtained as
\begin{eqnarray}
\beta _{{\rm{sinr}}} = \tfrac{{{\gamma _{_{\rm{SR}}}}}}{{{\gamma _{_{\rm{SR}}}}\mathcal{A}_r|h_0|^2 + \sqrt {{\gamma _{_{\rm{SR}}}}({\gamma _{_{\rm{SR}}}} + 1){{\gamma }_{_{\rm{RD}}}}\mathcal{A}_r|h_0|^2} }}, \label{eq:beta_opt}
\end{eqnarray}
which also satisfies the non-oscillatory condition in \eqref{eq:beta2<1/f2}. Substituting \eqref{eq:beta_opt} into \eqref{eq:SINR}, we can express the e-SINR as
\begin{eqnarray}
{\gamma _{{\rm{sinr}}}} = \tfrac{{{\gamma _{_{\rm{SR}}}}{{\gamma }_{_{\rm{RD}}}}}}{{{\gamma _{_{\rm{SR}}}}\mathcal{A}_r|h_0|^2 + {{ \gamma }_{_{\rm{RD}}}} + 2\sqrt {{\gamma _{_{\rm{SR}}}}({\gamma _{_{\rm{SR}}}} + 1){{\gamma }_{_{\rm{RD}}}}\mathcal{A}_r|h_0|^2} }}.  \label{eq:SINR_opt}
\end{eqnarray}
Obviously, fully utilizing the relay-harvested energy is not prerequisite in designing the relay gain $\beta _{\rm sinr}$. For example, the redundant energy can be harvested by the relay in addition to the necessary energy harvested to support the relay gain $\beta _{\rm sinr}$. To extend the relaying transmission time as long as possible, the TS factor is designed as small as possible such that the relay-harvested energy is just enough to implement the relay gain $\beta _{\rm sinr}$.
Therefore, by solving $\beta _{\rm sinr} = \beta _{\max}$ for any $\alpha$ ($0<\alpha <1$), we can provide the TS factor for the SINR relay as
\begin{eqnarray}
{\alpha _{{\rm{sinr}}}} = \tfrac{{\sqrt {{\gamma _{_{\rm{SR}}}} + 1} }}{{\sqrt {{\gamma _{_{\rm{SR}}}} + 1}  + \eta \sqrt {{\gamma _{_{\rm{SR}}}}{\gamma _{_{\rm{RD}}}}\mathcal{A}_r|h_0|^2} }}. \label{eq:alpha_opt}
\end{eqnarray}
Since no redundant energy has been harvested by employing $\alpha _{\rm sinr}$, the relaying transmission time $(1-\alpha _{\rm sinr})T$ is longer than those of other $\alpha$s satisfying $\alpha > {\alpha _{{\rm{sinr}}}}$.
The TS factor $\alpha _{\rm sinr}$ can be computed at the destination, or at the relay locally when the relay can access the global CSI.

\begin{ppro}
The outage probability achieved by the SINR relay is given by
\begin{eqnarray}
\!\!\!\!\!\!\!\!\!\!\!\!\!\! {P_{{\rm{out}}}} &\!\!\!\!=\!\!\!\!& 1 \!-\! \frac{1}{{{\lambda _1}}}\int\limits_{\frac{{d_1^m{\gamma _{{\rm{th}}}}\sigma _r^2}}{{\mathcal{L}_1{P_s}}}}^\infty  {\tfrac{{1 + K}}{{1 + K + \mathcal{A}_r \sigma _0^2\rho }}{e^{ - \frac{z}{{{\lambda _1}}} - \frac{{K \mathcal{A}_r\sigma _0^2\rho }}{{1 + K + \mathcal{A}_r\sigma _0^2\rho }}}}} {\rm{d}}z,  \label{eq:P_out_opt}
\end{eqnarray}
where $\rho  \triangleq \frac{{{\gamma _{{\rm{th}}}} + \frac{{{\mathcal{L}_1 }{P_s}z(1 + 2{\gamma _{{\rm{th}}}})}}{{d_1^m\sigma _r^2}} + 2\sqrt {{\gamma _{{\rm{th}}}}(1 + {\gamma _{{\rm{th}}}})\frac{{{\mathcal{L}_1 }{P_s}z}}{{d_1^m\sigma _r^2}}\left( {1 + \frac{{{\mathcal{L}_1 }{P_s}z}}{{d_1^m\sigma _r^2}}} \right)} }}{{\frac{{{\mathcal{L}_2 }{\lambda _2}{{({\mathcal{L}_1 P_s}z - d_1^m\sigma _r^2{\gamma _{{\rm{th}}}})}^2}}}{{d_1^{2m}d_2^m\sigma _r^2\sigma _d^2{\gamma _{{\rm{th}}}}}}}}$.
\end{ppro}
\begin{proof}
A proof is provided in Appendix C.
\end{proof}

\begin{ppro}
The ergodic capacity achieved by the SINR relay is given by\footnote{The numerical computation of this expression can be quickly performed in softwares such as Matlab and Mathematica.}
\begin{eqnarray}
{C_{_{\rm E}}} \!=\! - \int\limits_0^\infty \!\! {\int\limits_0^\infty \!\! {f_{_{|h_0|^2}}}(w)  \rho    {{\log }_2}(1 + \gamma )} {\rm{d}}w{\rm{d}}\gamma,
\label{eq:ce_opt1}
\end{eqnarray}
where $\rho  = v{e^{ - \frac{{2d_1^m\sigma _r^2\gamma }}{{{P_s}{\lambda _1}}}}}\left( {\frac{{2d_1^m\sigma _r^2{K_1}(v)}}{{{P_s}{\lambda _1}}} + \frac{{v{K_0}(v)}}{{2\gamma (\gamma  + 1 - \sqrt {\gamma (\gamma  + 1)} )}}} \right)$ and  $v  \triangleq 2\sqrt {\frac{{\mathcal{A}_r w d_1^md_2^m\sigma _d^2\gamma \left( {2\gamma  + 1 + 2\sqrt {\gamma(\gamma + 1)} } \right)}}{{{\mathcal{L}_1} {\mathcal{L}_2}{P_s}{\lambda _1}{\lambda _2}}}}  $.
\end{ppro}
\begin{proof}
A proof is provided in Appendix D.
\end{proof}

Although the TS factor does not appear directly in the expressions in Proposition 3 and Proposition 4,
the instantaneous CSI-based $\alpha_{\rm sinr}$ in \eqref{eq:alpha_opt} has been employed in the derivations of $P_{\rm out}$ and $C_{_{\rm E}}$. Therefore, the average throughputs in \eqref{eq:R_DL} and \eqref{eq:R_DT} for the SINR relay
will be rewritten as
\begin{subequations}
\begin{eqnarray}
{R_{_{{\rm{DL}}}}}  &\!\!=\!\!& \left\{ {\begin{array}{*{20}{c}}
  {\mathbb{E}\{(1-\alpha_{\rm sinr})R\},}& \gamma_{\rm sinr} > \gamma_{\rm th} \\
  {0,}& \gamma_{\rm sinr} < \gamma_{\rm th}
\end{array}} \right.\!\!, \label{eq:R_DL_new}\\
{R_{_{\rm{DT}}}}  &\!\!=\!\!& \mathbb{E}\{(1 - \alpha_{\rm sinr} )\log _2(1+\gamma_{\rm sinr})\}. \label{eq:R_DT_new}
\end{eqnarray}
\end{subequations}

\subsection{Target Relay}

When the SINR relay  is employed, exact knowledge of the channel SNR $\gamma _{_{\rm RD}}$ has to be exploited, which can be estimated only at the destination. A feedback channel is then required if $\alpha _{\rm sinr}$ is computed locally at the relay. Therefore, this subsection proposes a simplified relay control scheme, i.e., the target relay, that aims to achieve a target e-SINR $\hat \gamma $. To avoid using knowledge of $\gamma _{_{\rm RD}}$, the target e-SINR should satisfy $\hat \gamma < \gamma _{_{\rm SR}}$. Denoting the e-SINR achieved by the target relay as $\gamma _{\rm tar}$ ($\gamma _{\rm tar} = \hat \gamma$), the target relay is designed to maximize its e-SINR, i.e., to design the relay gain $\beta _{\rm tar}$ such that $\gamma _{\rm tar} = \gamma _{\rm sinr}$.
Denoting the TS factor for the target relay by $ \alpha _{\rm tar}$ and substituting $\mu  _{\rm tar} \triangleq \frac{{ \alpha _{\rm tar}  \eta }}{{1 -   \alpha _{\rm tar}}}$ into  \eqref{eq:beta_max} and \eqref{eq:SINR_max}, we can write the relay gain and  e-SINR as
\begin{eqnarray}
\beta _{\rm tar } = \tfrac{{ \mu  _{\rm tar}{\gamma _{_{\rm{SR}}}}}}{{1+{\gamma _{_{\rm{SR}}}} +  { \mu _{\rm tar} {\gamma _{_{\rm{SR}}}} \mathcal{A}_r |h_0|^2 }}}  \label{eq:beta_tar}
\end{eqnarray}
and
\begin{eqnarray}
{\gamma _{\rm tar}} = \tfrac{{\mu _{\rm tar} {\gamma _{_{\rm{SR}}}}{{\gamma }_{_{\rm{RD}}}}}}{{{\gamma _{_{\rm{SR}}}} + \left( {\mu _{\rm tar} {\gamma _{_{\rm{SR}}}}\mathcal{A}_r|h_0|^2 + 1} \right)\left( { \mu _{\rm tar} {{ \gamma }_{_{\rm{RD}}}} + 1} \right)}}, \label{eq:SINR_tar}
\end{eqnarray}
respectively. For a given $\hat \gamma$, by eliminating $\gamma _{_{\rm RD}}$ from the equation pair $\{ \gamma _{\rm tar} = \hat \gamma, \gamma _{\rm sinr} = \hat \gamma \}$, the TS factor is given by ${\alpha _{{\rm{tar}}}} =$
\begin{eqnarray}
\tfrac{{({\gamma _{_{\rm{SR}}}} + 1)({\gamma _{_{\rm{SR}}}} - \hat \gamma )}}{{({\gamma _{_{\rm{SR}}}} + 1)({\gamma _{_{\rm{SR}}}} - \hat \gamma  + \eta \hat \gamma {\gamma _{_{\rm{SR}}}}\mathcal{A}_r |h_0|^2) + \eta {\gamma _{_{\rm{SR}}}}\mathcal{A}_r|h_0|^2\sqrt {\hat \gamma (\hat \gamma  + 1){\gamma _{_{\rm{SR}}}}({\gamma _{_{\rm{SR}}}} + 1)} }}. \nonumber
\end{eqnarray}
Since $\alpha _{\rm tar}$ harvests only a necessary energy to support the relay gain $\beta _{\rm tar}$, the relaying transmission time $(1-\alpha _{\rm tar})T$ is longer than those of other $\alpha$s satisfying $\alpha > \alpha _{\rm tar}$.  Also, since $\hat \gamma < \gamma _{_{\rm SR}}$, $\alpha _{\rm tar} $ satisfies $0< \alpha _{\rm tar} <1$.
When $\hat \gamma \ge \gamma _{_{\rm SR}}$, we have $\alpha _{{\rm{tar}}} \le 0$. In this case, no time is assigned for EH such that information relaying fails due to a lack of power.
Alternatively, we can reset TS factor by $\alpha _{\rm tar} =1$ such that only EH is implemented for the entire time block.

\begin{ppro}
The outage probability achieved by the target relay is given by
\begin{eqnarray}
\!\!\!\!\!\!\!\!\!\!\!\!{P_{{\rm{out}}}} &\!\!\!\!= \!\!\!\!& 1 - \frac{1}{{{\lambda _1}}} \!\!\!\!\!\!\!\!\! \int\limits_{\frac{{d_1^m\gamma _{{\rm{th}}}^2\sigma _r^2(\hat \gamma  + 1)}}{{\mathcal{L} {P_s}(\hat \gamma  + 2\hat \gamma {\gamma _{{\rm{th}}}} - \gamma _{{\rm{th}}}^2)}}}^\infty  \!\!\!\!\!\!\!\!\! {\tfrac{{1 + K}}{{1 + K - \mathcal{A}_r \sigma _0^2\rho }}{e^{ - \frac{z}{{{\lambda _1}}} + \frac{{K \mathcal{A}_r  \sigma _0^2\rho }}{{1 + K - \mathcal{A}_r  \sigma _0^2\rho }}}}} {\rm{d}}z, \label{eq:P_out_tar}
\end{eqnarray}
where $\rho  \triangleq \frac{{d_1^md_2^m{\gamma _{{\text{th}}}}\sigma _d^2\omega }}{{\mathcal{L}_2 {\lambda _2}((\omega  - 1){\gamma _{{\rm{th}}}} - 1)((\omega  - 1)\mathcal{L}_1 {P_s}z - d_1^m\sigma _r^2)}}$ and $\omega  \triangleq \sqrt {\frac{{(\hat \gamma  + 1)(d_1^m\sigma _r^2 + \mathcal{L}_1 {P_s}z)}}{{\mathcal{L}_1 {P_s}\hat \gamma z}}} $.
\end{ppro}
\begin{proof}
A proof of this proposition is similar to the proof of Proposition 3.
\end{proof}


For the target relay, the effective EH and relaying transmission fail when $\gamma _{_{\rm SR}} \le \hat \gamma $. In this case, the e-SINR $\gamma _{\rm tar}$ does not exist because of an impractical $\mu _{\rm tar}$. Although the derivative of $P_{\rm out}$ in Proposition 5 can still be obtained by the mathematical manipulation, it cannot be used to represent the PDF of $\gamma _{\rm tar}$ because of the discontinuity of $\gamma _{\rm tar}$. Therefore, finding the PDF of $\gamma _{\rm tar}$ to evaluate $C_{_{\rm E}} = \mathbb E \{ \log _2(1+\gamma _{\rm tar}) \}$ is difficult. As an alternative, the ergodic capacity achieved by the target relay will be investigated by simulations. Similar to the SINR relay, the average throughputs for the delay-limited and delay-tolerant transmissions are given by
\begin{subequations}
\begin{eqnarray}
{R_{_{{\rm{DL}}}}}  &\!\!=\!\!& \left\{ {\begin{array}{*{20}{c}}
  {\mathbb{E}\{(1-\alpha_{\rm tar})R\},}& \gamma_{\rm tar} > \gamma_{\rm th} \\
  {0,}& \gamma_{\rm tar} < \gamma_{\rm th}
\end{array}} \right.\!\!, \label{eq:R_DL_new2}\\
{R_{_{\rm{DT}}}}  &\!\!=\!\!& \mathbb{E}\{(1 - \alpha_{\rm tar} )\log _2(1+\gamma_{\rm tar})\}. \label{eq:R_DT_new2}
\end{eqnarray}
\end{subequations}

\subsection{Direct Transmission}

Due to an energy loss resulted from the EH circuit and the two-hop radio propagation, SWIET FDR suffers from a degradation in the EE. However, since SWIET FDR provides a relative freedom besides the source-destination direct transmission, it results in an additional performance advantage. In this subsection, we consider the source-destination transmission scenario with a direct link, where an SWIET FDR node can be employed to assist the information transfer.

For the transmission scheme exploring both the direct link and SWIET FDR, we assume that the signal transmitted from the source directly and the delayed replica transmitted from the relay are fully resolvable by the destination, so that they can be appropriately co-phased and merged via maximum ratio combing (MRC).
The received e-SINR at the destination can be expressed as
\begin{eqnarray}
\gamma _{_{\rm MRC}} \triangleq \gamma _{_{\rm SD}} +  \gamma _{_{{\rm SRD}}},
\end{eqnarray}
where $\gamma _{_{\rm SD}} \triangleq \tfrac{\mathcal{L}_1 P_s |h_3|^2} {d_3^m \sigma _d^2}$ is the SNR achieved by the direct transmission, $h_3$ is the channel coefficient of the direct link, and $d_3$ is the distance between the source and destination.
Then, the instantaneous throughput of the direct transmissions with SWIET FDR can be expressed as
\begin{eqnarray}
R_{_{\rm I}} = \alpha  \log _2(1+ \gamma_{_{\rm SD}}) + (1-\alpha  ) \log _2 (1+  \gamma _{_{\rm MRC}} ).
\end{eqnarray}
It is easy to show that the outage probability of the direct transmission without SWIET FDR is given by
\begin{eqnarray}
P_{\rm out}^{_{(\rm SD)}} &=&  1- e^{-\tfrac{d_3^m \sigma_d^2 \gamma _{\rm th}}{\mathcal{L}_1  P_s \lambda _3}},
\end{eqnarray}
where $\lambda_3$ is the mean of the exponential random variable $|h_3|^2$.
\begin{cor}
The outage probability of the direct transmission with SWIET FDR is upper-bounded by
\begin{eqnarray}
P_{\rm out, ub}^{_{(\rm MRC)}} = P_{\rm out}^{_{(\rm SD)}}   P_{\rm out}^{_{_{({\rm SRD})}}},
\end{eqnarray}
where the outage probabilities $P_{\rm out}^{_{(\rm SRD)}}$s achieved by the maximum relay, SINR relay, and target relay are given by  \eqref{eq:proposition1}, \eqref{eq:P_out_opt}, and \eqref{eq:P_out_tar}, respectively.
\end{cor}
\begin{proof}
A proof is provided in Appendix E.
\end{proof}

Based on Corollary 1, the average throughput of the direct transmissions with the maximum relay in the delay-limited transmission mode is lower bounded by
\begin{eqnarray}
R_{_{\rm DL, lb}}^{_{(\rm SD+SRD)}}  =  \alpha   (1 - P_{\rm out}^{_{(\rm SD)}})R  +   (1 - \alpha )(1 - P_{\rm out, ub}^{_{(\rm MRC)}}) R.
\end{eqnarray}
For the direct transmission with the SINR relay and target relay, where $\alpha$ depends on the instantaneous CSI, the average throughput can be evaluated by
\begin{eqnarray}
R_{_{\rm DL}}^{_{\rm SD+SRD}} = \mathbb{E} \{c_1\alpha   R + c_2(1-\alpha )R \}, \label{eq:R_DL_new3}
\end{eqnarray}
where $c_1 \!=\! \left\{ {\begin{array}{*{20}{c}}
  {1,}& \gamma_{_{\rm SD}} > \gamma _{\rm th} \\
  {0,}& \gamma_{_{\rm SD}} < \gamma _{\rm th}
\end{array}} \right.$
and $c_2 \!=\! \left\{ {\begin{array}{*{20}{c}}
  {1,}& \gamma_{_{\rm MRC}} > \gamma _{\rm th} \\
  {0,}& \gamma_{_{\rm MRC}} < \gamma _{\rm th}
\end{array}} \right.$\!.

In the delay tolerant transmission mode, the ergodic capacity of the direct transmission without SWIET FDR can be expressed as $C_{_{\rm E}}^{_{(\rm SD)}} = \frac{1}{\ln 2} e^{\tfrac{d_3^m \sigma_d^2}{\mathcal{L}_1 P_s \lambda _3}} \Gamma \left(0, \tfrac{d_3^m \sigma_d^2}{\mathcal{L}_1 P_s \lambda _3}\right)$, where $\Gamma(\cdot, \cdot)$ in the upper incomplete gamma function \cite[Eq. (8.350.2)]{Table_Integrals}.
The throughput of the direct transmission with SWIET FDR can be expressed as
\begin{eqnarray}
R_{_{\rm DT}}^{_{(\rm SD+SRD)}}    = \mathbb E \{\alpha \} C_{_{\rm E}}^{_{(\rm SD)}}  + \mathbb E \{(1-\alpha ) \log_2(1+ \gamma_{_{\rm MRC}}) \}. \label{eq:R_DT_new3}
\end{eqnarray}
When only the statistical CSI is available, \eqref{eq:R_DT_new3} can be rewritten as $R_{_{\rm DT}}^{_{(\rm SD+SRD)}}    = \alpha C_{_{\rm E}}^{_{(\rm SD)}}  + (1-\alpha ) C_{_{\rm E}}^{_{(\rm MRC)}}$, where $C_{_{\rm E}}^{_{(\rm MRC)}}  \triangleq \log_2(1+ \gamma_{_{\rm MRC}})$ is the ergodic capacity achieved by the MRC at the destination. Due to mathematical intractability, the throughputs expressed in \eqref{eq:R_DL_new3} and \eqref{eq:R_DT_new3} and the related EE will be numerically investigated in the next section.

\subsection{CSI-related Issue}

Based on the pilot symbols sent from the source over dual-hop channels, CSI can be estimated to facilitate the wireless information and energy transfer \cite{SWIPT_architecture, MIMO_B_SWIPT, SWIPT_FD_selfenergy, SWIPT_protocol_AF, SWIPT_protocol_DF}. Similarly to the works of \cite{path_selection} and \cite{Energy-Efficient}, a request-to-send (RTS)/clear-to-send (CTS) based channel estimation scheme can be employed at the beginning of the entire transmission and instantaneous CSI can be obtained. Surely, RTS/CTS based channel estimation incurs extra overhead and energy consumption.

Moreover, the CSI estimation error may degrade the performance of SWIET FDR system. For the maximum relay, the bisection method needs full CSI, i.e., $|h_i|^2$ for $i=0, 1$, and 2 to evaluate the throughput. To compute $\alpha$, the SINR relay also needs full CSI, while the target relay only needs $|h_0|^2$ and $|h_1|^2$. The CSI estimation error can be modeled as ${\tilde h}_i  = h_i + h_{i,e}$ for $i=0, 1, 2$, and 3, respectively, where the estimation error $h_{i,e}$ is a complex Gaussian random variable independent of $h_i$. In addition, we assume that $h_{i,e}$ has the zero mean and the variance $\kappa |h_i|^2$, where the scaling factor $\kappa$ is a positive real number representing the relative ratio between the estimation error and true CSI. In simulations, we apply the above procedure to investigate the effect of CSI estimation error on the performance metrics.

\section{Numerical Results}

\begin{table}[tb] 
\caption{Simulation Parameters}
\renewcommand\arraystretch{1.4}
\begin{center}
\renewcommand{\arraystretch}{1.1}
\setlength\tabcolsep{4pt}
\begin{tabular}{|l|l|l|}
\hline
No. & Parameter & Value  \\
\hline
1 & Carrier frequency                        & 915 MHz   \\
2 & Bandwidth                                & 200 kHz \\
3 & Fixed transmission rate $R$              & 2 bps/Hz \\
4 & Path loss at the reference distance 1 m: $\mathcal{L}$  & -30 dB \\
5 & Distance between the nodes: $d_1$ and $d_2$ & 10 m \\
6 & Path loss exponent $m$                  & 3 ~~~  \\
7 & Means of dual-hop channel gains: $\lambda_1$ and $\lambda_2$ & 1~~~~~ \\
8 & Noise power: $\sigma _r^2$ and $\sigma _d^2$ & $-$95 dBm \\
9 & Source transmit antenna gain            & 18 dBi  \\
10 & Relay transmit antenna gain            & 8 dBi \\
11 & Rician $K$-factor of the RSI channel    & 6 dB \\
12 & Energy harvesting efficiency $\eta$     & 0.8~~~~  \\
\hline
\end{tabular}
\end{center}
\end{table}

This section presents some numerical results to validate the analytical expressions developed in the previous section and discuss the EE performances for the considered relay control schemes. Unless otherwise stated, the parameters used in simulations are given in Table 1.


\begin{figure}[htbp]
\begin{center}
\subfigure[e-SINR versus $P_s$.]{\includegraphics[width=2.95in]{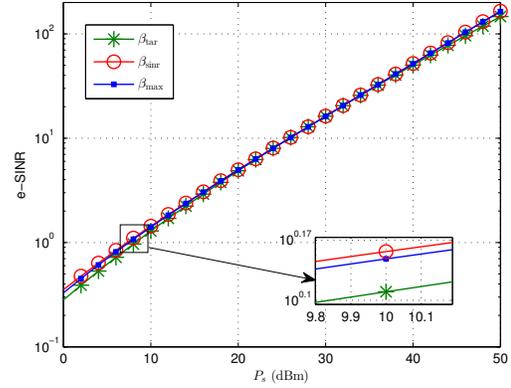}  \label{fig:2a} }
\vspace{-0.2in}
\subfigure[$\alpha$ versus $P_s$]{\includegraphics[width=2.95in]{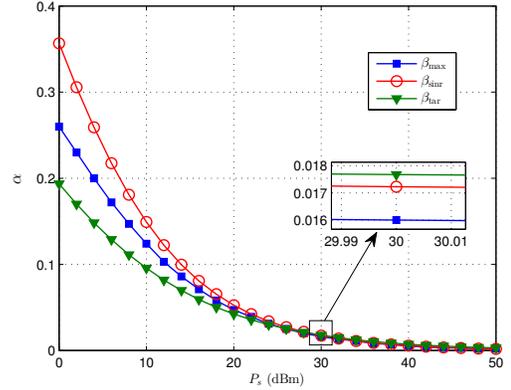} \label{fig:2b} }
\vspace{0.15in}
\caption{Performance metrics in instantaneous transmission.}
\end{center}
\vspace{-0.22in} \label{fig:11}
\end{figure}

Fig. 2 investigates the e-SINR and corresponding $\alpha$ versus the source transmission power in the instantaneous transmission. In Fig. 2, we focus on a single frame with the following channel setting: $|h_0|^2=0.342$, $|h_1|^2=1.898$, and $|h_2|^2=0.986$. When the source transmission power $P_s$ increases from $0$ dBm to $50$ dBm, the target e-SINRs for the target relay are set to increase linearly from $0$ dB to $20$ dB with a step size of $0.8$ dB.
As observed, the three relay control schemes achieve almost the same e-SINR and the highest e-SINR is achieved by the SINR relay. In the region of the low transmission power, the e-SINR achieved by the target relay is slightly lower than that of the maximum relay. When $P_s$ increases, the e-SINR gap between the maximum relay and target relay becomes negligible. Fig. 2(b) shows that TS factor decreases for the three relay control schemes when $P_s$ increases. This result shows that the relay node can harvest an enough energy in a short time when $P_s$ is large enough ($P_s > 20$ dBm). In addition, since the SINR relay applies the largest $\alpha$ to harvest energy, this leads to the shortest relaying transmission time. From the point of view of the EE, Fig. 2 shows that the SINR relay achieves its EE depending on the e-SINR more than the effective relaying transmission time, whereas the target relay achieves its EE with its target e-SINR and corresponding $(1-\alpha_{\rm tar})T$. Note that the e-SINR and TS factor of~ the maximum relay are~ obtained by the bisection searching, which will result in the highest EE in instantaneous transmission, as depicted in Fig. 4.

\begin{figure}[htbp]
\begin{center}
\includegraphics[width=2.95in]{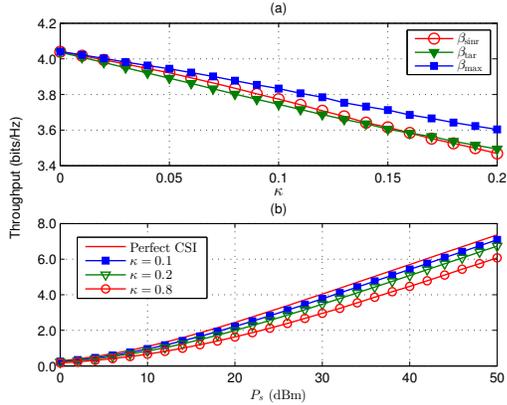}
\vspace{-0.15in}
\caption{The instantaneous throughput under CSI estimation error.}
\end{center}
\vspace{-0.05in}
\label{fig:5}
\end{figure}


The impacts of CSI estimation error on the instantaneous throughput are illustrated in Fig. 3, where the channel realization is as same as that of Fig. 2. In Fig. 3(a), we set $P_s = 30$ dBm, $\hat \gamma = 12$ dB, and assume that $h_0$, $h_1$, and $h_2$ suffer the same level of CSI estimation error. The curves in Fig. 3 are obtained by averaging over $20,000$ random CSI estimation errors. As observed in Fig. 3(a), the throughputs of all the relay control schemes decrease when $\kappa$ increases. When $\kappa = 0.1$, the throughputs of the maximum relay and target relay become about $0.20$ bps/Hz lower than that of perfect CSI ($\kappa = 0$), whereas the throughput of the SINR relay decreases about $0.27$ bps/Hz from that of the perfect CSI. This result indicates that the SINR relay is more likely to be affected by the CSI estimation error than the maximum relay. For the SINR relay, Fig. 3(b) plots the throughput versus $P_s$ under different $\kappa$s. As observed in Fig. 3(b), the throughput of the SINR relay decreases slightly for $\kappa =0.1$ and $\kappa=0.2$. When $P_s$ increases, the throughput degradation also increases. Moreover, Fig. 3(b) also depicts the throughput degrading for $\kappa = 0.8$, which is corresponding to the scenario with a serious CSI estimation error. When $P_s = 30$ dBm, the throughput degradation is about $1$ bps/Hz for $\kappa =0.8$. This result shows that the SINR relay can not work well in the case of a serious CSI estimation error.


\begin{figure}[htbp]
\begin{center}
\subfigure[EE versus $|h_0|^2$.]{\includegraphics[width=2.95in]{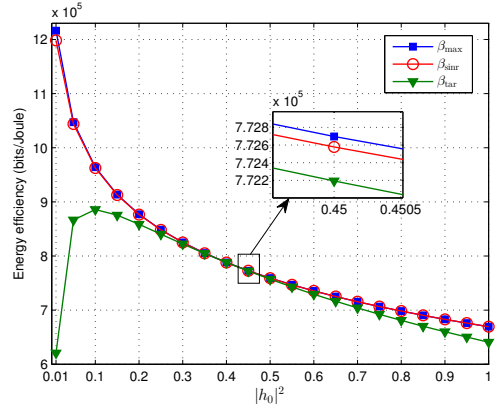}  \label{fig:4a} }
\vspace{-0.2in}
\subfigure[EE versus $P_s$]{\includegraphics[width=2.95in]{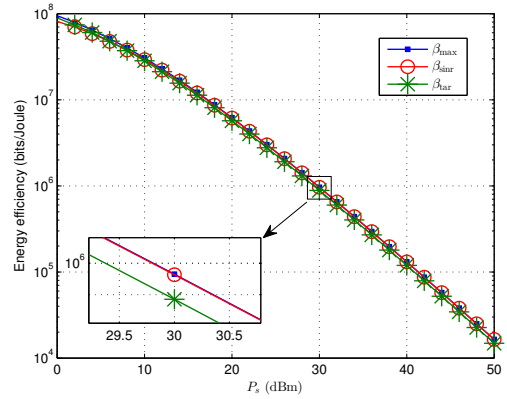} \label{fig:4b} }
\vspace{0.15in}
\caption{EE in instantaneous transmission.}
\end{center}
\label{fig:11}
\vspace{-0.22in}
\end{figure}

Fig. 4 illustrates the impacts of $|h_0|^2$ and $P_s$ on the EE of instantaneous transmission. In Fig. 4(a), we set $P_s=30$ dBm, $\hat \gamma =12$ dB, $h_1$ and $h_2$ are the same as those of Fig. 2. As a result, the EE achieved by the SINR relay is very close to that of the maximum relay. When $|h_0|^2$ increases, the EEs of the maximum relay and SINR relay decreases simultaneously. For the target relay with $\hat \gamma =12$ dB, its EE approaches to that of the maximum relay in the region of $|h_0|^2$ between $0.2$ and $0.8$. Beyond this region, the EE of the target relay degrades slightly compared to that of the maximum relay.
The curves of EE versus $P_s$ are plotted in Fig. 4(b), where we fixed $|h_0|^2=0.1$. For the target relay, $\hat \gamma$ increase from $0$ dB to $20$ dB with a step size of $0.8$ dB. In the considered whole region of $P_s$, the EEs of all the three relay control schemes are at the same level. In addition, the EEs of all the three relay control schemes decrease when $P_s$ increases. Actually, the maximum EEs are achieved by the three relay control schemes when $P_s$ is less than $-8$ dBm for this channel realization. Since the EE maximization is a concave fractional program of $P_s$, the EE becomes a monotonically decreasing function of the considered practical $P_s$ ($P_s>0$ dBm). The results in Fig. 4 also show that the target relay is competitive compared to the maximum relay and SINR relay since the target relay has the closed-form TS factor and does not need the second-hop CSI.

Fig. 5 shows the impact of $\hat \gamma $ on the EE for the target relay under the same channel realization of Fig. 4(b). As observed, the maximum throughput can always be achieved by the target relay. For example, when $P_s=30$ dBm, the channel SNR of the first-hop link is $\gamma _{\rm SR} =85.78$ dB, whereas the target e-SINR that achieves the maximum EE is about $\hat \gamma =15$ dB, which satisfies $\hat \gamma < \gamma _{\rm SR}$ (Note that we set $\hat \gamma = 12 $ dB when $P_s=30$ dBm in Fig. 4(b)). Moreover, an approximately $10^5$ bits/Joule EE decreases when $\hat \gamma$ moves by approximately $4$ dB away from $\hat \gamma =15$ dB. This EE degradation is about $10\%$ (within the same order of magnitude) to the EE achieved by the maximum relay. When the second-hop CSI is not available, the target relay is suitable for scenarios in which the priority is EE maximization.

\begin{figure}[htbp]
\begin{center}
\includegraphics[width=2.95in]{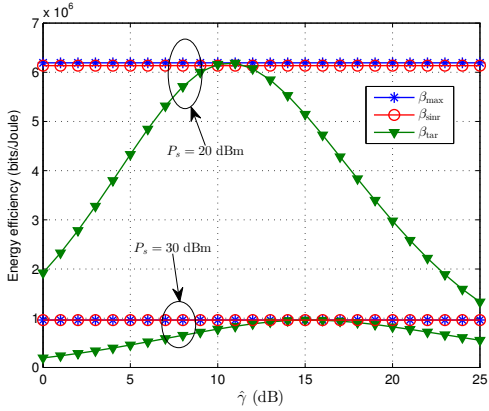}
\vspace{-0.15in}
\caption{EE versus $\hat \gamma$.}
\end{center}
\vspace{-0.2in} \label{fig:5}
\end{figure}


The outage probability versus the transmission power and average RSI channel gain is depicted in Fig. 6. In Fig. 6(a), we set $\sigma_0^2 = 0.1$ and $\hat \gamma$ increases from $0$ dB to $15$ dB with a step size of $0.6$ dB. Fig. 6(a) shows that the analytical and simulation results match well. As observed in Fig. 6(a), the outage probability gap between the maximum relay and SINR relay becomes larger when $P_s$ increases. The reason for this phenomenon is that the time of energy harvesting of the SINR relay is a little more longer  than that of the maximum relay, so that the corresponding e-SINR is more larger than that of the maximum relay in the region of high $P_s$. Note that the TS factor of the SINR relay is just a little larger than that of the maximum relay, which will not degrade the EE dramatically, as verified by the results of Fig. 7. In Fig. 6(b), we set $P_s = 30$ dBm and $\hat \gamma = 8$ dB. As observed in Fig. 6(b), the outage probability of the maximum relay varies very slowly when $\sigma_0^2$ is smaller than $0.1$.
However, when the value of $\sigma_0^2$ is smaller than $0.1$, the outage probability of the SINR relay and target relay decreases dramatically with a decreasing $\sigma_0^2$. Moreover, Fig. 6(b) also verifies that the SINR relay achieves the smallest outage probability.


\begin{figure}[tbp]
\vspace{-0in}
\begin{center}
\subfigure[Outage probability versus $P_s$. ]{\includegraphics[width=2.95in]{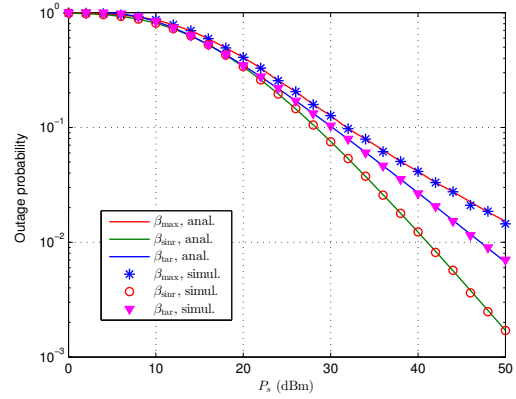}  \label{fig:6a} }
\vspace{-0.15in}
\subfigure[Outage probability versus $\sigma_0^2$. ]{\includegraphics[width=2.95in]{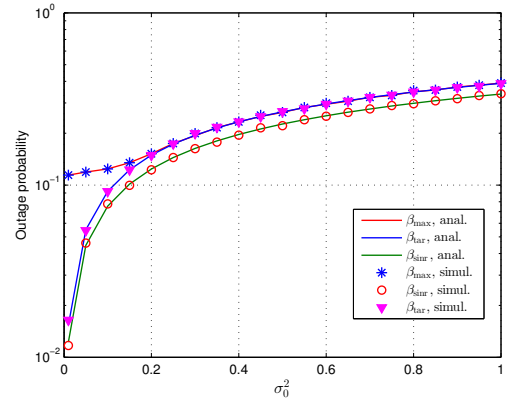} \label{fig:6b} }
\vspace{0.18in}
\caption{Outage probability in delay-limited transmission.}
\end{center}
\label{fig:6}
\vspace{-0.22in}
\end{figure}

Fig. 7 shows the EEs of all the three relay control schemes for delay-limited transmission, where $P_s$, $\sigma_0^2$, and $\hat \gamma$ are the same as those of Fig. 6. As a result, the SINR relay achieves a highest EE in the whole region of $P_s$. This result verifies that a larger e-SINR and a proper $\alpha_{\rm sinr} $ of the SINR relay outperforms the maximum relay in both the throughput and EE. Fig. 7(a) also shows that the EE gap between the SINR relay and maximum relay in the low region of $P_s$ is larger  than that in the high region of $P_s$. In addition, in the low region of $P_s$, the EEs of all the relay control schemes increase with the increasing of $P_s$, whereas when $P_s$ reaches a threshold, the EE begins to decrease with the increasing of $P_s$. As observed in Fig. 7(a), the larger EEs are obtained in the low region of $P_s$, e.g., from $5$ dBm to $15$ dBm.
Thus, a relatively high outage probability achieved in the low region of $P_s$ becomes preferable from the point view of EE. For example, Fig. 6(a) shows that the SINR relay achieves $P_{\rm out} = 0.8$ when $P_s = 10$ dB, whereas Fig. 7(a) shows that the SINR relay achieves the highest EE when $P_s = 10$ dB. Although the throughput is a monotonically non-decreasing function of $P_s$ (not shown in Fig. 7), the results of Fig. 7(a) reveal a different changing trend for the EE versus $P_s$.
Fig. 7(b) investigates the EE versus $\sigma_0^2$ with a middle transmission power, i.e., $P_s=30$ dBm. As depicted in Fig. 7(b), the SINR relay achieves the highest EE in a large region of $\sigma_0^2$. Only when $\sigma_0^2<0.03$, the EE of the SINR relay is just a little lower than that of the maximum relay.
Fig. 7 also shows that the target relay achieves a competitive EE performance. Although the EE of the target relay is lower than that of the maximum relay, the EE gap between the maximum relay and target relay decreases with the increasing of $P_s$.

\begin{figure}[htbp]
\begin{center}
\subfigure[EE versus $P_s$.]{\includegraphics[width=2.95in]{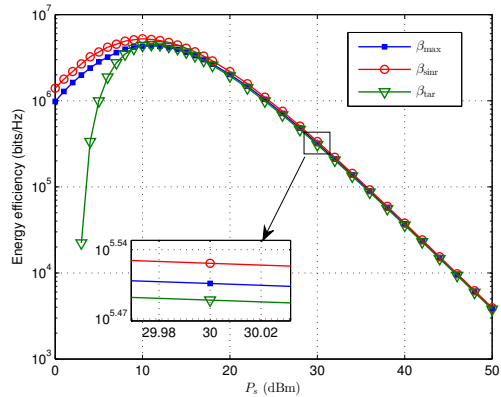}  \label{fig:7a} }
\vspace{-0.15in}
\subfigure[EE versus $\sigma_0^2$]{\includegraphics[width=2.95in]{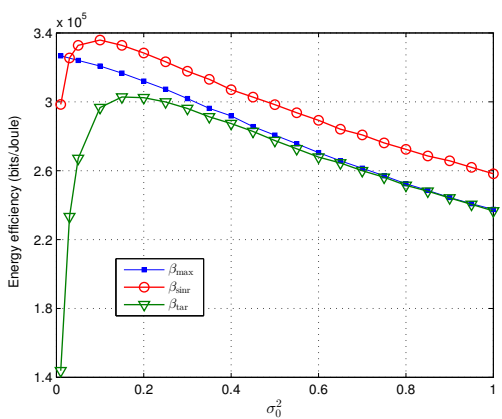} \label{fig:7b} }
\vspace{0.18in}
\caption{EE in delay-limited transmission.}
\end{center}
\label{fig:7}
\vspace{-0.22in}
\end{figure}

The ergodic capacities are examined in Fig. 8, where $\sigma_0^2 = 0.4$, $P_s$ and $\hat \gamma$ are set as same as those of Fig. 7. For the ergodic capacities of the maximum relay and SINR relay, the analytical expressions match well with the simulation results. The ergodic capacity achieved by the SINR relay is the highest among all the three relay control schems since the SINR relay always achieves the highest e-SINR. Note that the simulated ergodic capacity of the maximum relay is achieved with the bisection-obtained $\alpha$ and only statistics of CSI have been employed in the bisection searching. Fig. 8 also shows that the target relay achieves a competitive ergodic capacity compared with the maximum relay.

Fig. 9 examines the EEs for the delay-tolerant transmission. In Fig. 9, $P_s$, $\sigma_0^2$, and $\hat \gamma$ are set as same as those of Fig. 8. Fig. 9(a) shows that the EE gap between the maximum relay and SINR relay is very small. In the middle and high regions of $P_s$ ($P_s>4.5$ dBm), the SINR relay achieves an EE higher than that of the maximum relay. Moreover, in the low region of $P_s$ ($P_s < 4.5$ dBm), the SINR relay achieves an EE lower than that of the maximum relay. Fig. 9(a) also shows that the EEs of all the three relay control schemes monotonically increase with the decreasing of $P_s$ for the considered region of $P_s$. Note that the maximum EEs of all the three relay control schemes are obtained by further decreasing $P_s$ into the extremely small region (around $-20$ dBm), so that such a small transmission power is out of practical use. Fig. 9(b) shows that the EE of the maximum relay and SINR relay decrease with the increasing of $\sigma_0^2$. ~For the considered target e-SINR, Fig. 9 also shows that  the target relay achieves a competitive EE compared with the maximum relay.

\begin{figure}[htbp]
\begin{center}
\subfigure[Ergodic capacity versus $P_s$. ]{\includegraphics[width=2.95in]{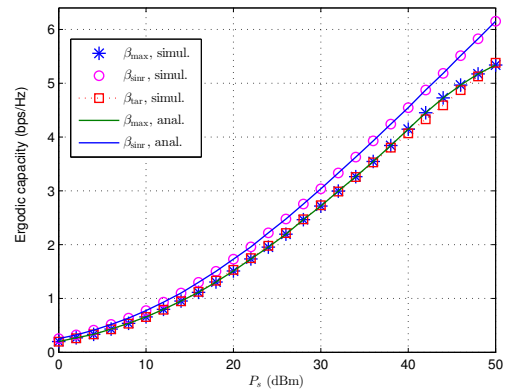}  \label{fig:8a} }
\vspace{-0.15in}
\subfigure[Ergodic capacity versus $\sigma_0^2$. ]{\includegraphics[width=2.95in]{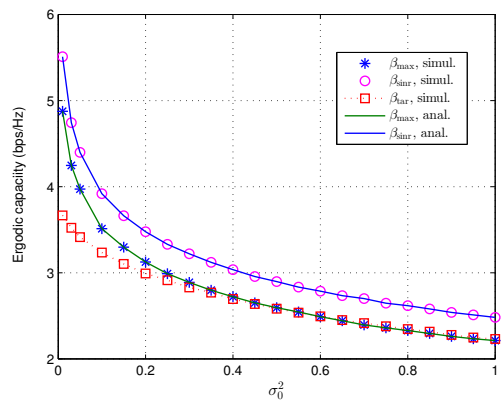} \label{fig:8b} }
\vspace{0.18in}
\caption{Ergodic capacity in delay-tolerant transmission.}
\end{center}
\label{fig:8}
\vspace{-0.22in}
\end{figure}

Fig. 10 investigates the impact of the relay position on the system performance. In Fig. 10, we set $P_s = 30$ dBm, $\sigma_0^2 = 0.1$, $d_3 = 20$ m, and $d_1 + d_2 = d_3$. As observed, the highest outage probabilities are achieved by all the three relay control schemes when $d_1/d_3 = 0.5$. When the relay moves away from the middle position, the outage probabilities decrease. For the delay-tolerant transmission, the lowest throughputs are obtained when the relay is placed midway. The reason for this is that the e-SINR satisfies $\gamma _{_{\rm SRD}} \propto (d_1 d_2)^{-m} $ and $\gamma _{_{\rm SRD}}$ achieves its smallest value when $d_1 = d_2$. Moreover, Fig. 10 shows that the lowest EE  is obtained when the relay is placed midway. Therefore, the relay should be placed closer to the source or destination to achieve a better system performance.


\begin{figure}[htbp]
\begin{center}
\subfigure[EE versus $P_s$.]{\includegraphics[width=2.95in]{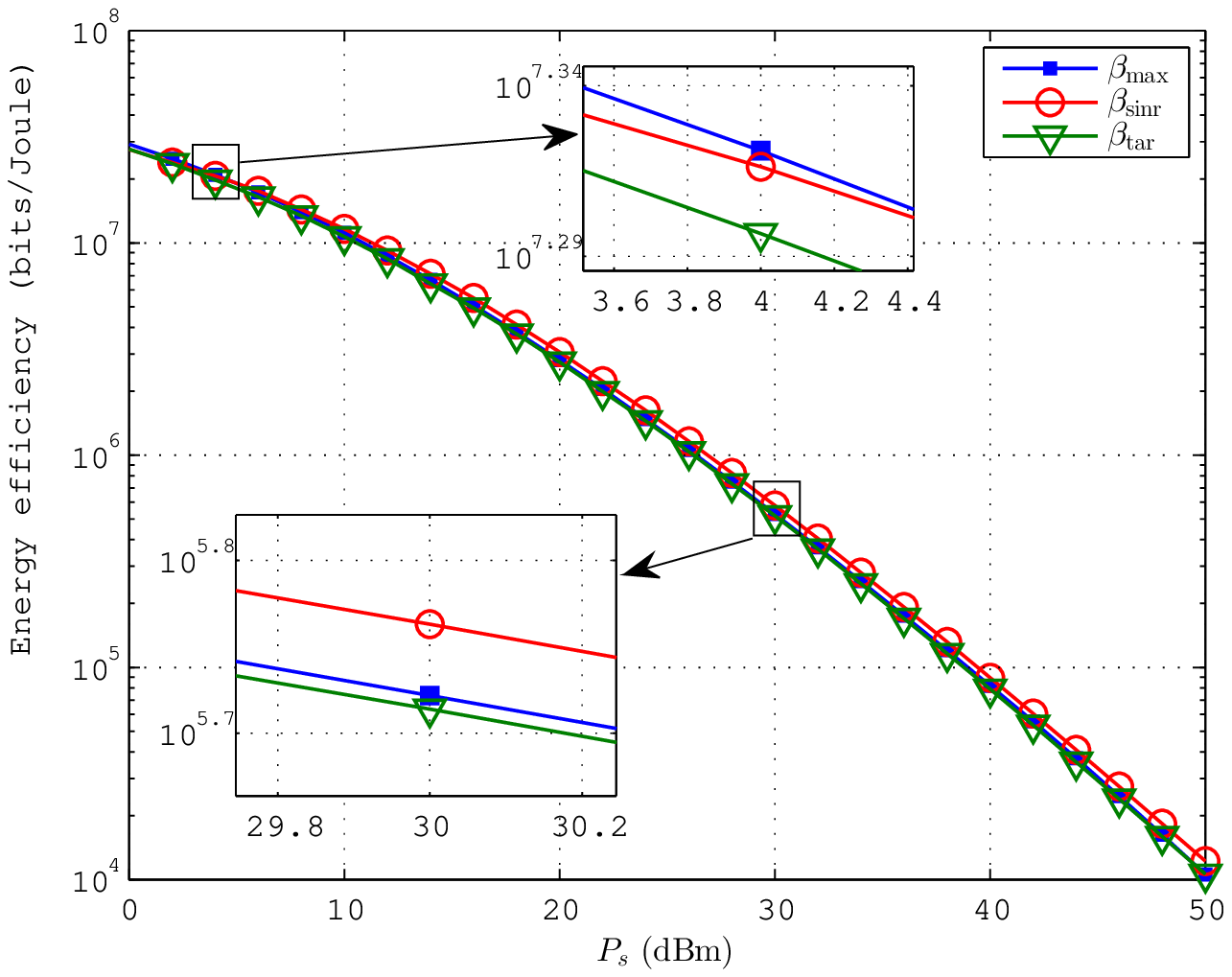}  \label{fig:9a} }
\vspace{-0.15in}
\subfigure[EE versus $\sigma_0^2$]{\includegraphics[width=2.95in]{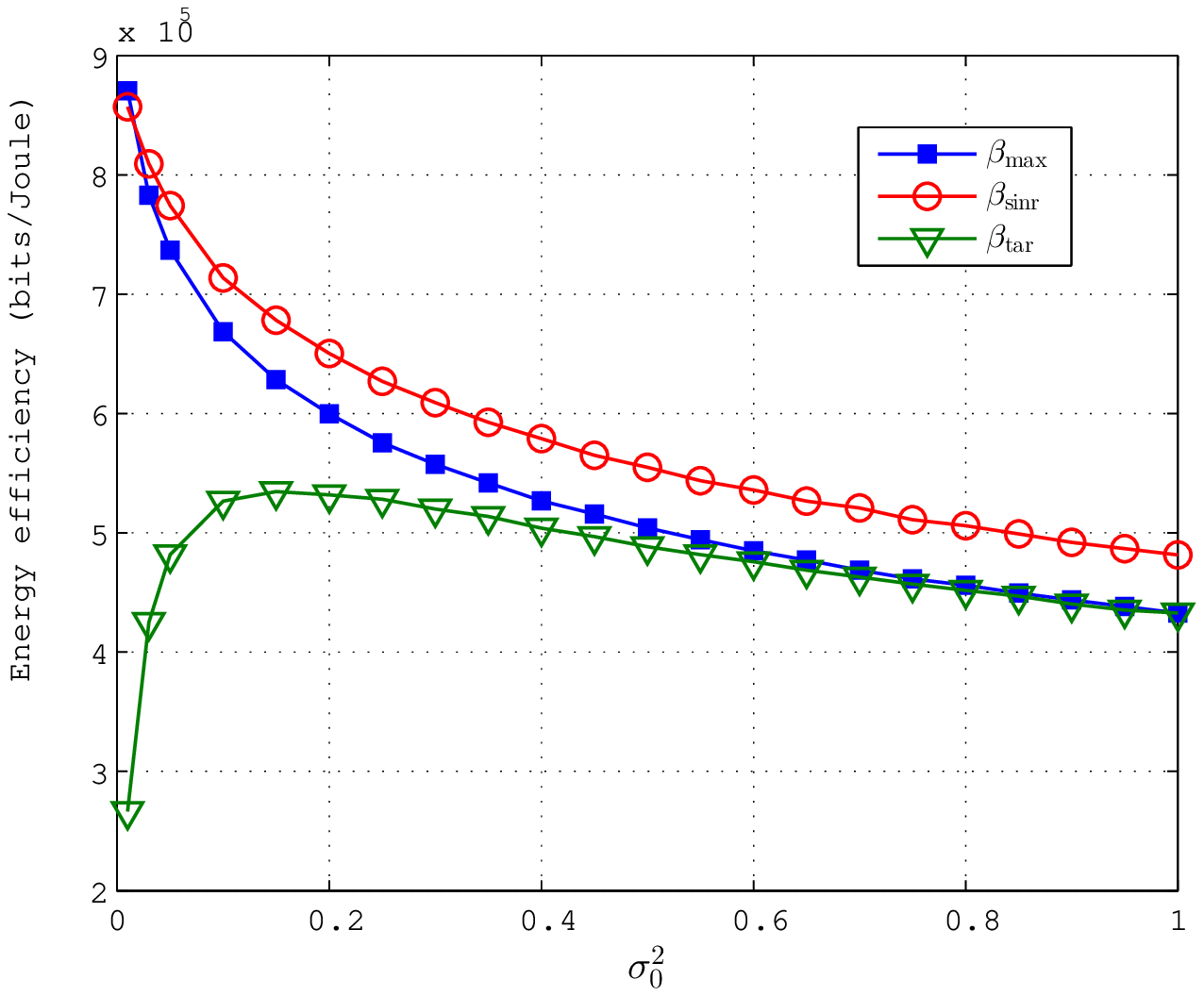} \label{fig:9b} }
\vspace{0.18in}
\caption{EE in delay-tolerant transmission.}
\end{center}
\label{fig:9}
\vspace{-0.22in}
\end{figure}

Fig. 11 compares the system performances of the direct transmission with and without SWIET FDR in the delay-limited transmission mode. In Fig. 11, we set $\sigma_0^2=0.1$, $d_1 = d_2 =10$ m, $d_3=20$ m, and consider $R = 2$ bps/Hz and $6$ bps/Hz, respectively. Since the EEs of all the three relay control schemes are at the same order of magnitude, we consider only the SINR relay to clearly show the curves. As observed in Fig. 11(a), the outage probability becomes very small when the direct link is available. In the low region of $P_s$, the outage probability of the direct transmission with SWIET FDR is almost the same as that of direct transmission. After $P_s$ increases to a threshold value, the outage probability of the direct transmission with SWIET FDR can be further decreased compared to that of the direct transmission. However, when $R$ increases from $2$ bps/Hz to $6$ bps/Hz, the outage probability decreasing requires a higher $P_s$, which results in a great EE decreasing, as depicted in Fig. 11(b). Since the outage probability of the direct transmission is already very small, Fig. 11(a) suggests that SWIET FDR is benefit in assisting the low-rate and high quality-of-service transmission. For example, when $P_s = 20$ dBm, the outage probability of the direct transmission is about $1.2 \times 10^{-6}$ for $R = 2$ bps/Hz, which can be further decreased to $4.0 \times 10^{-7}$ with the assistance of the SWIET FDR. Another important observation is that the EE of the direct transmission (with and without SWIET FDR) is monotonically decreasing when $P_s$ increases. As a result, the EE gap between the direct transmission and relaying transmission decreases when $P_s$ increases. For $R=2$ bps/Hz, the EE gap almost disappears when $P_s$ surpasses $30$ dBm. The reason for this phenomenon is that the contributions of different small outage probabilities (e.g., $P_{\rm out}<0.1$) to the EE become relatively small when these outage probabilities are obtained by the relatively large $P_s$s.
Therefore, it can be concluded that the SWIET FDR is an energy-effective scheme for the low-rate transmission when the direct link is not available.

\begin{figure}[htbp]
\begin{center}
\includegraphics[width=2.95in]{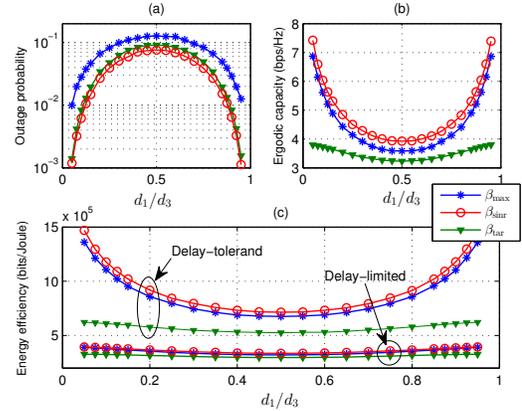}
\vspace{-0.15in}
\caption{Impact of relay position on system performance.}
\end{center}
\vspace{-0.15in}\label{fig:7}
\end{figure}

Fig. 12 compares the system performances of the direct transmission with and without SWIET FDR in the delay-tolerant transmission mode. In Fig. 12, the simulation parameters are the same as those of Fig. 11. As observed, the ergodic capacity of the direct transmission is much higher than that of relaying transmission. When $P_s = 30$ dBm, the ergodic capacity gap between the direct transmission and relaying transmission is about $17.7$ bps/Hz. More over, when SWIET FDR is employed to assist the direct transmission, the ergodic capacity almost keeps the same as that of the direct transmission, since the e-SINR enhancement due to SWIET FDR is relatively small compared to that of direct transmission. As a result, the EE of the relaying system assisted direct transmission is almost the same as that of the direct transmission. Thus, the performance enhancement for the delay-tolerant direct transmission by employing SWIET FDR is trivial.

\section{Conclusion}

This paper has studied three relay control schemes to improve EE of SWIET FDR systems. The EE maximization problem has been formulated as a concave fractional program of source transmission power, so that the three relay control schemes are designed separately. Under Rician fading RSI channel, analytical expressions of outage probability and ergodic capacity have been derived for the maximum relay, SINR relay, and target relay.
For the maximum relay, the EE maximization problem is shown to be a concave function with respect to TS factor, so that the optimized TS factor can be obtained by the bisection method. By employing instantaneous CSI, the SINR relay and target relay with collateral TS factors have gained advantages in improving EE. In delay-limited and delay-tolerant transmissions, numerical results show that the SINR relay achieves a better performance than that of the maximum relay in terms of outage probability, ergodic capacity, and EE, whereas the target relay achieves a competitive system performance without requiring CSI of the second-hop. It has shown that a worse outage performance may result in a higher EE in delay-limited transmission. In low-rate delay-limited transmissions, numerical results also show that SWIET FDR can achieve a competitive EE compared to source-destination direct transmission. When the relay is placed midway between the source and relay, the worst EE is achieved, as well as the outage probability and erogdic capacity.

\begin{figure}[htbp]
\begin{center}
\subfigure[Outage probability versus $P_s$. ]{\includegraphics[width=2.95in]{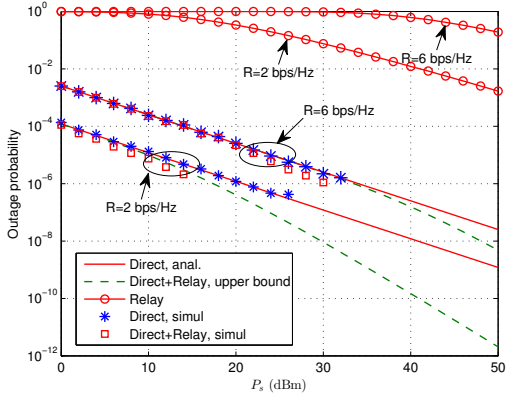}  \label{fig:11a} }
\vspace{-0.15in}
\subfigure[EE versus $P_s$. ]{\includegraphics[width=2.95in]{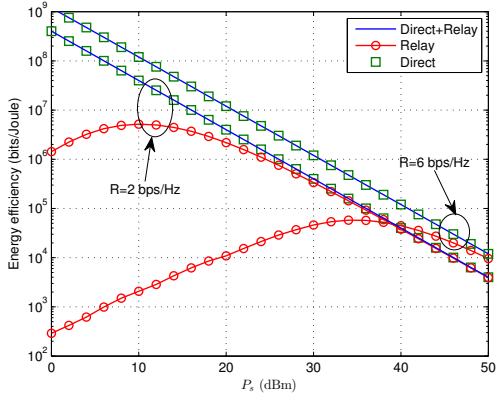} \label{fig:11b} }
\vspace{0.18in}
\caption{Performance in delay-limited direct transmission.}
\end{center}
\label{fig:11}
\vspace{-0.22in}
\end{figure}

\begin{figure}[htbp]
\begin{center}
\includegraphics[width=2.95in]{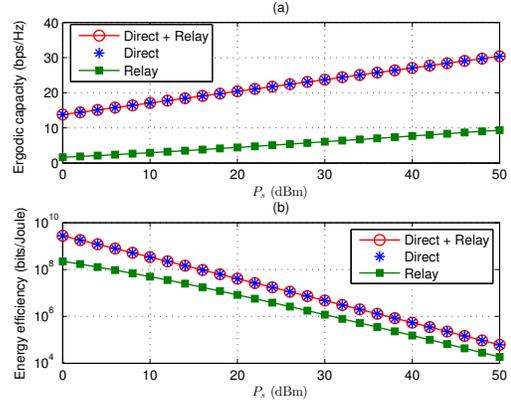}
\vspace{-0.15in}
\caption{Performance in delay-tolerant direct transmission.}
\end{center}
\vspace{-0.22in}\label{fig:7}
\end{figure}

\section*{Appendix A: A proof of Proposition 1}
\renewcommand{\theequation}{A.\arabic{equation}}
\setcounter{equation}{0}
Substituting \eqref{eq:SINR_max} into $P_{\rm out} = \Pr(\gamma _{\max}<\gamma _{\rm th})$, the outage probability is given by
\begin{eqnarray}
{P_{\rm out}} &\!\!\!\!=\!\!\!\!& \Pr \left( {\tfrac{{\mu {\gamma _{_{\rm SR}}}{{\gamma }_{_{\rm RD}}}}}{{{\gamma _{_{\rm SR}}} + \left( {\mu {\gamma _{\rm SR}}\mathcal{A}_r|h_0|^2 + 1} \right)\left( {\mu {{\gamma }_{\rm RD}} + 1} \right)}} < {\gamma _{\rm th} }} \right) \nonumber \\
&\!\!\!\!=\!\!\!\!& \Pr \left( {|h_2|^2 < \tfrac{{\bar a|h_1|^2 + b}}{{\bar c|h_1|^4 - d|h_1|^2}}} \right),
\end{eqnarray}
where $\bar a \triangleq \mathcal{L}_1 {P_s}d_1^md_2^m\sigma _d^2{\gamma _{\rm th} }(1 + \mu \mathcal{A}_r |h_0|^2)$, $b \triangleq d_1^{2m}d_2^m\sigma _r^2\sigma _d^2{\gamma _{\rm th} }$, $\bar c \triangleq \mathcal{L}_1^2 \mathcal{L}_2  P_s^2\mu (1 - \mu {\gamma _{\rm th} }\mathcal{A}_r|h_0|^2)$, and $d \triangleq \mathcal{L}_1 \mathcal{L}_2 {P_s}d_1^m\sigma _r^2\mu {\gamma  _{\rm th}}$. Given that the term $\bar c|h_1|^4-d|h_1|^2$ can be positive or negative and $|h_2|^2$ is always greater than a negative number, $P_{\rm out}$ can be simplified as
\begin{eqnarray}
{P_{{\rm{out}}}}  &\!\!\!\!\!\!=\!\!\!\!\!\!& \left\{\!\!\!\! {\begin{array}{*{20}{c}}
{\Pr \!\! \left( \!\! {|h_2|^2 \!<\! \frac{{\bar a|h_1|^2 + b}}{{\bar c|h_1|^4 - d|h_1|^2}}} \!\!\right),}\\
{1,}\\
{1,}
\end{array}} \right.\!\!\!\!\! \! \begin{array}{*{20}{c}}
{|h_0|^2 < {{{1} \over {\mu {\gamma _{{\rm{th}}}}}}}~{\rm and}~|h_1|^2 > \frac{d}{c}}\\
{|h_0|^2 < {{{1} \over {\mu {\gamma _{{\rm{th}}}}}}}~{\rm and}~|h_1|^2 < \frac{d}{c}}\\
{|h_0|^2 > {{{1} \over {\mu {\gamma _{{\rm{th}}}}}}} ~{\rm and}~|h_1|^2 > 0}
\end{array}\!\!\!.   \label{ap:P_out1} \nonumber
\end{eqnarray}
Denote the PDF of $|h_1|^2$ and the cumulative distribution function (CDF) $|h_2|^2$ by
${f}_{_{|h_1|^2}}(z)\triangleq \frac{1}{\lambda_1} e^{-\frac{z}{\lambda _1}}$   and $F_{_{|h_2|^2}}(z) \triangleq \Pr (|h_2|^2<z) = 1-e^{-\frac{z}{\lambda _{g}}}$, respectively, where $\lambda_1$ and $\lambda_2$ are the means of $|h_1|^2$ and $|h_2|^2$, respectively. Then, $P_{\rm out}$ can be evaluated as
\begin{eqnarray}
\!\!{P_{{\rm{out}}}} &\!\!\!\!\!=\!\!\!\!\!& \int\limits_{w=\frac{{1 }}{{\mu {\gamma _{{\rm{th}}}}}}}^\infty  {\int\limits_{z=0}^\infty  {{{f}_{_{|h_0|^2}}}(w){{f}_{_{|h_1|^2}}}(z)} } {\rm{d}}z{\rm{d}}w \nonumber \\
\!\!\!&\!\!\!\!\!\!\!\!\!\!& + \int\limits_{w=0}^{\frac{{1 }}{{\mu {\gamma _{{\rm{th}}}}}}} {\int\limits_{z = 0}^{\frac{d}{c}} {{{f}_{_{|h_0|^2}}}} (w){{f}_{_{|h_1|^2}}}(z){\rm{d}}z{\rm{d}}w} \nonumber \\
\!\!\!&\!\!\!\!\!\!\!\!\!\!& + \int\limits_{w=0}^{\frac{{1}}{{\mu {\gamma _{{\rm{th}}}}}}} {\int\limits_{z = \frac{d}{c}}^\infty  {{{f}_{_{|h_0|^2}}}} (w){{f}_{_{|h_1|^2}}}(z)\Pr \!\! \left( {|h_2|^2 < \tfrac{{az + b}}{{c{z^2} - dz}}} \right)\!{\rm{d}}z} {\rm{d}}w \nonumber \\
\!\!\!&\!\!\!\!\!=\!\!\!\!\!& 1 \!-\! \tfrac{1}{\lambda _1}\int\limits_{w=0}^{\frac{{1}}{{\mu {\gamma _{{\rm{th}}}}}}} {\int\limits_{z = \frac{d}{c}}^\infty \!\! {{ {f}_{_{|h_0|^2}}(w) e^{ - \left( {\frac{z}{{{\lambda _1}}} + \frac{{az + b}}{{(c{z^2} - dz){\lambda _2}}}} \right)}} {\rm{d}}z} } {\rm{d}}w, \label{ap:P_out}
\end{eqnarray}
where $a \triangleq \mathcal{L}_1 {P_s}d_1^md_2^m\sigma _d^2{\gamma _{\rm th} }(1 + \mu \mathcal{A}_r w)$, $c \triangleq \mathcal{L}_1^2 \mathcal{L}_2 P_s^2\mu (1 - \mu {\gamma _{\rm th} } \mathcal{A}_r w)$, and ${f}_{_{|h_0|^2}}(w)$ is the PDF of $|h_0|^2$ given by \eqref{eq:pdf_f}.
The analytical expression for $P_{\rm out}$ in \eqref{ap:P_out} cannot be further simplified. However,
at the high SNRs, $\frac{{az + b}}{{c{z^2} - dz}}$ has an approximation as $\frac{{a }}{{c{z} - d}}$, so that $P_{\rm out}$ can be approximated as
\begin{eqnarray}
\!{P_{{\rm{out}}}} &\!\!\!\!\!\!\!\!\!\!\!\! \approx \!\!\!\!\!&\! 1  - \frac{1}{{{\lambda _1}}}\int\limits_{w=0}^{{{{1} \over {\mu {\gamma _{{\rm{th}}}}}}}} {\int\limits_{z = {{d \over c}}}^\infty  {{ {{{f}_{_{|h_0|^2}}}} (w) e^{ - \left( {\frac{z}{{{\lambda _1}}} + \frac{a}{{(cz - d){\lambda _2}}}} \right)}} {\rm{d}}z} } {\rm{d}}w \nonumber \\
\!&\!\!\!\!\!\!\!\!\!\!\!\! \mathop  = \limits^{x \triangleq cz - d} \!\!\!\!\!&\! 1 - \frac{1}{{{c \lambda _1}}}\int\limits_{w=0}^{\frac{{1}}{{\mu {\gamma _{{\rm{th}}}}}}} {\int\limits_{x = 0}^\infty  {{ {{{f}_{_{|h_0|^2}}}} (w) e^{ - \left( {\frac{x}{{c{\lambda _1}}} + \frac{a}{{x{\lambda _2}}}} \right)}}{e^{ - \frac{d}{{c{\lambda _1}}}}} {\rm{d}}x} } {\rm{d}}w  \nonumber \\
\!&\!\!\!\!\!\!\!\!\!\!\!\!=\!\!\!\!\!&\! 1 -  \int\limits_{0}^{{{{1 } \over {\mu {\gamma _{{\rm{th}}}}}}}} { {{{f}_{_{|h_0|^2}}}} (w) \rho {K_1}(\rho ){e^{ - \frac{d}{{c{\lambda _1}}} }}} {\rm{d}}w, \label{ap:P_out_app}
\end{eqnarray}
where $\rho \triangleq \sqrt{\frac{4a}{c\lambda _1 \lambda _2}} $ and $K_1(\cdot)$ is the first-order modified Bessel function of the second kind \cite[Eq. (8.432)]{Table_Integrals}. The last equality in \eqref{ap:P_out_app} is obtained by applying $\int _0^\infty  {{e^{ - \frac{\alpha }{{4x}} - \beta x}}dx}  = \sqrt {\frac{\alpha }{\beta }} {K_1}(\sqrt {\alpha \beta } )$ \cite[Eq. (3.324.1)]{Table_Integrals}.

\section*{Appendix B: A proof of Proposition 2}
\renewcommand{\theequation}{B.\arabic{equation}}
\setcounter{equation}{0}
In the RSI dominated scenario, the ergodic capacity achieved by the maximum relay can be expressed as
\begin{eqnarray}
{C_{_{\rm E}}} &=& \mathbb{E}\{ {\log _2}(1 + \mu {\gamma _{_{\rm{RD}}}} + \mu ^2 \gamma _{_{\rm{RD}}} \mathcal{A}_r |h_0|^2 )\}  \nonumber \\
& & - \mathbb{E}\{ {\log _2}(1 + \mu ^2 \gamma _{_{\rm{RD}}} \mathcal{A}_r |h_0|^2 ) \} .  \label{ap:C_E_max}
\end{eqnarray}
Define $x \triangleq |h_1|^2|h_2|^2$, where the PDF of $x$ is given by $f(x) = \frac{2}{{{\lambda _1}{\lambda _2}}}{K_0}\left( {2\sqrt {{{x \over {{\lambda _1}{\lambda _2}}}}} } \right)$. Given $|h_0|^2$, the second term in the right hand side (RHS) of \eqref{ap:C_E_max} can be evaluated as
\begin{eqnarray}
\mathbb{E}\{ {\log _2}(1 +  \mu ^2 \gamma _{_{\rm{RD}}}\mathcal{A}_r |h_0|^2 ) |~ |h_0|^2 \} & & ~~~~~~~~~~~~~~~~~~~~~~~~~~~~~ \nonumber
\end{eqnarray}
\vspace{-0.25in}
\begin{eqnarray}
\!\!&\!\!\!\!\!=\!\!\!\!\! &  \tfrac{2}{{{\lambda _1}{\lambda _2}\ln 2}}\!\!\int\limits_0^\infty   {\ln \left( {1 + \tfrac{{\mathcal{L}_1 \mathcal{L}_2 \mu^2 {P_s}\mathcal{A}_r|h_0|^2 x}}{{d_1^m d_2^m \sigma _d^2}}} \right)\!{K_0}\!\left( {2\sqrt {\tfrac{x}{{{\lambda _1}{\lambda _2}}}} } \right){\rm{d}}x} \nonumber \\
\!\!&\!\!\!\!\!\mathop  = \limits^{({\rm a})}\!\!\!\!\!& \tfrac{2}{{{\lambda _1}{\lambda _2}\ln 2}}\!\!\int\limits_0^\infty  {G_{2,2}^{1,2}\left( {\left. {\tfrac{{\mathcal{L}_1 \mathcal{L}_2 \mu^2 {P_s} \mathcal{A}_r |h_0|^2 x}}{{d_1^md_2^m\sigma _d^2}}} \right|{}_{1,0}^{1,1}} \right)\!{K_0}\!\left( {2\sqrt {\tfrac{x}{{{\lambda _1}{\lambda _2}}}} } \right){\rm{d}}x} \nonumber \\
\!\!&\!\!\!\!\!\mathop  = \limits^{({\rm b})} \!\!\!\!\!& \tfrac{1}{{\ln 2}}G_{4,2}^{1,4}\left( {\left. {\tfrac{{\mathcal{L}_1 \mathcal{L}_2  \mu^2 {P_s}{\lambda _1}{\lambda _2} \mathcal{A}_r |h_0|^2}}{{d_1^md_2^m\sigma _d^2}}} \right|{}_{1,0}^{0,0,1,1}} \right), \label{ap:C_E_max_item1}
\end{eqnarray}
where $G_{m,n}^{p,q}(x)$ is the Meijer G-function \cite[Eq. (9.301)]{Table_Integrals}.
In \eqref{ap:C_E_max_item1}, we have used the relationship \cite[Eq. (8.4.6.5)]{PBM:90:Book:v3} in the step (a) and the integral identity \cite[Eq. (7.821.3)]{Table_Integrals} in the step (b), respectively.
Then, the second term in the RHS of \eqref{ap:C_E_max} can be expressed as
\begin{eqnarray}
\mathbb{E}\{ {\log _2}(1 +  \mu ^2 \gamma _{_{\rm{RD}}} \mathcal{A}_r |h_0|^2 ) \} & & ~~~~~~~~~~~~~~~~~~~~~~~~~~~~~ \nonumber
\end{eqnarray}
\vspace{-0.25in}
\begin{eqnarray}
&\!\!\!\!\!=\!\!\!\!\!& \tfrac{1}{{\ln 2}} \int\limits_0^\infty f_{_{|h_0|^2}}(t) G_{4,2}^{1,4}\left( {\left. {\tfrac{{\mathcal{L}_1 \mathcal{L}_2 \mu^2 {P_s}{\lambda _1}{\lambda _2} \mathcal{A}_r t}}{{d_1^md_2^m\sigma _d^2}}} \right|{}_{1,0}^{0,0,1,1}} \right) {\rm d}t \nonumber \\
&\!\!\!\!\! \mathop  = \limits^{(\rm a)} \!\!\!\!\!&
\tfrac{e^{-K}}{K \ln{2}}\sum\limits_{n = 0}^\infty  {\tfrac{{{{\left( {\frac{{d_1^md_2^m\sigma _d^2K(K + 1)}}{{\mathcal{L}_1 \mathcal{L}_2 {\mu ^2}{P_s}{\lambda _1}{\lambda _2} \mathcal{A}_r \sigma _0^2}}} \right)}^{n + 1}}}}{{{{(n!)}^2}}}} \nonumber \\
&\!\!\!\!\! \!\!\!\!\!& \times
 G_{1,4}^{4,1}\left( {\left. {{\textstyle{{d_1^md_2^m\sigma _d^2(K + 1)} \over {\mathcal{L}_1 \mathcal{L}_2 {\mu ^2}{P_s}{\lambda _1}{\lambda _2} \mathcal{A}_r \sigma _0^2}}}} \right|{}_{0, - 1 - n, - 1 - n, - n}^{ - 1 - n}} \right),
\end{eqnarray}
where the infinite-series representation of $I_0(z)$ \cite[Eq. (8.447.1)]{Table_Integrals} has been applied in the step (a).
Now, the first term in the RHS of \eqref{ap:C_E_max} can be shown similarly as
\begin{eqnarray}
\mathbb{E}\{ {\log _2}(1 +  \mu \gamma _{_{\rm{RD}}} +  \mu ^2 \gamma _{_{\rm{RD}}} \mathcal{A}_r |h_0|^2 ) \} & & ~~~~~~~~~~~~~~~~~~~~~~~~ \nonumber
\end{eqnarray}
\vspace{-0.25in}
\begin{eqnarray}
&\!\!\!\!\!=\!\!\!\!\!& \tfrac{1}{{\ln 2}} \! \int\limits_0^\infty \! f_{_{|h_0|^2}}(w) G_{4,2}^{1,4}\left( {\left. {\tfrac{{\mathcal{L}_1 \mathcal{L}_2  \mu {P_s}{\lambda _1}{\lambda _2} (1+\mu \mathcal{A}_r w)}}{{d_1^md_2^m\sigma _d^2}}} \right|{}_{1,0}^{0,0,1,1}} \right) {\rm d}w. \nonumber
\end{eqnarray}
Then, the desired result follows immediately.

\section*{Appendix C: A proof of Proposition 3}
\renewcommand{\theequation}{C.\arabic{equation}}
\setcounter{equation}{0}
Substituting \eqref{eq:SINR_opt} into $P_{\rm out} = \Pr(\gamma _{\rm sinr}<\gamma _{\rm th} )$, the outage probability is given by $P_{\rm out} = $
\begin{eqnarray}
 \left\{\!\!\!\! {\begin{array}{*{20}{c}}
  {\Pr \!\! \left( \! {|{h_2}{|^2} \!<\! \tfrac{{\bar a|{h_1}{|^2} + \bar b + \bar c\sqrt {|{h_1}{|^2}(\mathcal{L}_1{P_s}|{h_1}{|^2} + d)} }}{{\mathcal{L}_2{{(\mathcal{L}_1{P_s}|{h_1}{|^2} - {\gamma _{{\rm{th}}}}d)}^2}}}} \! \right)\!\!,}& \!\! {|{h_1}{|^2} \!>\! \tfrac{{d_1^m{\gamma _{{\rm{th}}}}\sigma _r^2}}{{\mathcal{L}_1{P_s}}}} \\
  {1,}& \!\! {|{h_1}{|^2} \!<\! \frac{{d_1^m{\gamma _{{\rm{th}}}}\sigma _r^2}}{{\mathcal{L}_1{P_s}}}}
\end{array}} \right.\!\!\!\!,  \nonumber \label{ap:P_out_opt}
\end{eqnarray}
where $\bar a \triangleq \mathcal{L}_1 {P_s} \mathcal{A}_r |h_0|^2d_1^md_2^m\sigma _d^2{\gamma  _{\rm th}} ( {1 + 2{\gamma _{\rm th} }} )$, $\bar b \triangleq  \mathcal{A}_r |h_0|^2d_1^{2m} d_2^m\sigma _r^2\sigma _d^2\gamma  _{\rm th}^2$, $\bar c = 2\mathcal{A}_r |h_0|^2d_1^md_2^m\sigma _d^2{\gamma _{\rm th} }$ $\sqrt {\mathcal{L}_1 {P_s}(1 + {\gamma _{\rm th} }){\gamma _{\rm th} }} $, and $d \triangleq d_1^{m}\sigma _r^2$.
Substituting ${f}_{_{|h_1|^2}}(z) \triangleq \frac{1}{\lambda_1} e^{-\frac{z}{\lambda _1}}$, ${f}_{_{|h_0|^2}} (z)$ of \eqref{eq:pdf_f}, and $F_{_{|h_2|^2}}(z) \triangleq $ $\Pr (|h_2|^2<z) = 1-e^{-\frac{z}{\lambda _{2}}}$ into the above expression of $P_{\rm out}$, the outage probability can be written as
\begin{eqnarray}
\!{P_{{\rm{out}}}} &\!\!\!\!\!=\!\!\!\!\!& 1 \!-\! \tfrac{1}{{{\lambda _1}}} \!\!\!\!\!\!\!\! \int\limits_{{\frac{{d_1^m{\gamma _{{\rm{th}}}}\sigma _r^2}}{{\mathcal{L}_1{P_s}}}}}^\infty \!\!\! {\int\limits _0^\infty \!\!  {{{f}_{_{|h_0|^2}} (w) e^{ - \frac{z}{{{\lambda _1}}} \!-\! \frac{{w\left( {az + b + c\sqrt {z(\mathcal{L}_1 {P_s}z + d)} } \right)}}{{{{\mathcal{L}_2 {\lambda _2} (\mathcal{L}_1{P_s}z - {\gamma _{{\rm{th}}}}d)}^2}}}}}{\rm{d}}w} } {\rm{d}}z \nonumber \\
\! &\!\!\!\!\!=\!\!\!\!\!& 1 \!-\! \tfrac{1}{{{\lambda _1}}} \!\!\!\! \int\limits_{\frac{{d_1^m{\gamma _{{\rm{th}}}}\sigma _r^2}}{{\mathcal{L}_1{P_s}}}}^\infty \!\!\!\! {\tfrac{{1 + K}}{{1 + K + \mathcal{A}_r \sigma _0^2\rho }}{e^{ - \frac{z}{{{\lambda _1}}} - \frac{{K\mathcal{A}_r\sigma _0^2\rho }}{{1 + K + \mathcal{A}_r\sigma _0^2\rho }}}}} {\rm{d}}z, \label{ap:P_out_opt_2}
\end{eqnarray}
where $a \triangleq \mathcal{L}_1 {P_s} d_1^md_2^m\sigma _d^2{\gamma  _{\rm th}} \mathcal{A}_r {1 + 2{\gamma _{\rm th} }} )$, $b \triangleq d_1^{2m}d_2^m\sigma _r^2\sigma _d^2$ $\gamma  _{\rm th}^2 \mathcal{A}_r $, $c \triangleq 2d_1^md_2^m\sigma _d^2{\gamma _{\rm th} \mathcal{A}_r }\sqrt {\mathcal{L}_1 {P_s}(1 + {\gamma _{\rm th} }){\gamma _{\rm th} }} $, and
\begin{eqnarray}
\rho  = \tfrac{{{\gamma _{{\rm{th}}}} + \frac{{{\mathcal{L}_1 }{P_s}z(1 + 2{\gamma _{{\rm{th}}}})}}{{d_1^m\sigma _r^2}} + 2\sqrt {{\gamma _{{\rm{th}}}}(1 + {\gamma _{{\rm{th}}}})\frac{{{\mathcal{L}_1}{P_s}z}}{{d_1^m\sigma _r^2}}\left( {1 + \frac{{{\mathcal{L}_1 }{P_s}z}}{{d_1^m\sigma _r^2}}} \right)} }}{{\frac{{{\mathcal{L}_2 }{\lambda _2}{{(\mathcal{L}_1 {P_s}z - d_1^m\sigma _r^2{\gamma _{{\rm{th}}}})}^2}}}{{d_1^{2m}d_2^m\sigma _r^2\sigma _d^2{\gamma _{{\rm{th}}}}}}}}.  \label{ap:rho}
\end{eqnarray}
This proves Proposition 3.

\section*{Appendix D: A proof of Proposition 4}
\renewcommand{\theequation}{D.\arabic{equation}}
\setcounter{equation}{0}

The ergodic capacity can be evaluated by ${C_{_{\rm E}}} = \int\nolimits_0^\infty  {{f_{{\gamma _{{\rm{sinr}}}}}}(\gamma ){{\log }_2}(1 + \gamma )} {\rm{d}}\gamma $, where ${f_{{\gamma _{{\rm{sinr}}}}}}(\gamma )$ is the PDF of $\gamma _{\rm sinr}$.
By replacing $\gamma _{\rm th}$ with $\gamma$ in \eqref{ap:P_out_opt_2}, the CDF of $\gamma _{\rm sinr}$ can be obtained. Then, the PDF of $\gamma _{\rm sinr}$ can be computed by ${f_{{\gamma _{{\rm{sinr}}}}}}(\gamma )= \tfrac{{\partial {F_{{\gamma _{{\rm{sinr}}}}}}(\gamma )}}{{\partial \gamma }}$. However, the complicated expression of $\rho$ in \eqref{ap:rho} leads to a huge expression for ${f_{{\gamma _{{\rm{sinr}}}}}}(\gamma )$ such that the numerical evaluation of $C_{_{\rm E}}$ becomes difficult.
Fortunately, a high SNR  approximation can be applied to simplify the expression. By using the similar procedure as in \eqref{ap:P_out_app}, the PDF of $\gamma _{\rm sinr}$ can be computed as
\begin{eqnarray}
{f_{{\gamma _{{\rm{sinr}}}}}}(\gamma ) =  - \int\nolimits_0^\infty {f_{|h_0|^2}}(w)  \rho {\rm{d}}w, \label{ap:C_E_sub}
\end{eqnarray}
where $\rho \triangleq \tfrac{{\partial  {{e^{ - \frac{{2d_1^m\sigma _r^2\gamma }}{{{P_s}{\lambda _1}}}}}v{K_1}(v)} }}{{\partial \gamma }}$ and $v$ is given by
\begin{eqnarray}
v  \triangleq   2 \sqrt {\tfrac{{\mathcal{A}_r wd_1^md_2^m\sigma _d^2\gamma \left( {2\gamma  + {\text{1}} + 2\sqrt {r(r + 1)} } \right)}}{{\mathcal{L}_1 \mathcal{L}_2 {P_s}{\lambda _1}{\lambda _2}}}}. \nonumber
\end{eqnarray}
Then, the ergodic capacity is given by
\begin{eqnarray}
{C_{\rm E}} \!=\! - \int\nolimits_0^\infty \!\! {\int\nolimits_0^\infty \!\! {f_{|h_0|^2}}(w) \rho {{\log }_2}(1 + \gamma )} {\rm{d}}w{\rm{d}}\gamma, \label{ap:ce_opt}
\end{eqnarray}
where the term $\rho$ can be further evaluated as
\begin{eqnarray}
  \rho =  - \frac{{2d_1^m\sigma _r^2v{K_1}(v)}}{{{P_s}{\lambda _1}}}{e^{ - \frac{{2d_1^m\sigma _r^2\gamma }}{{{P_s}{\lambda _1}}}}} + {e^{ - \frac{{2d_1^m\sigma _r^2\gamma }}{{{P_s}{\lambda _1}}}}}\tfrac{{\partial [v{K_1}(v)]}}{{\partial \gamma }}  \nonumber \\
   = v{e^{ - \tfrac{{2d_1^m\sigma _r^2\gamma }}{{{P_s}{\lambda _1}}}}}\left( {\tfrac{{2d_1^m\sigma _r^2{K_1}(v)}}{{{P_s}{\lambda _1}}} + \tfrac{{v{K_0}(v)}}{{2\gamma (\gamma  + 1 - \sqrt {\gamma (\gamma  + 1)} )}}} \right).
\end{eqnarray}

\section*{Appendix E: A proof of Corollary 1}
\renewcommand{\theequation}{E.\arabic{equation}}
\setcounter{equation}{0}

The outage probability of multiple relays assisted transmission can be written as
\begin{eqnarray}
P_{\rm out}^{_{(\rm SD+SRD)}}=\left\{  {\begin{array}{*{20}{c}}
  {\Pr   \left( {\gamma _{_{\rm{SD}}}} < {\gamma _{{\rm{th}}}} -   {\gamma _{_{{\rm{SRD}}}}} \right) ,}&{   {\gamma _{_{{\rm{SRD}} }}} < {\gamma _{{\rm{th}}}}} \\
  {0,}&{  {\gamma _{_{{\rm{SRD}}}}} > {\gamma _{{\rm{th}}}}}
\end{array}} \right.\!\!.  \label{ap:P_out_sd_srd_0}
\end{eqnarray}
Denoting the CDF and complementary CDF of $|h_3|^2$ by $F_{_{|h_3|^2}} (x)$ and $\bar F_{_{|h_3|^2}} (x)$, respectively, and denoting the PDF of $ {\gamma _{_{{\rm{SRD}} }}}$ by $f_{\gamma _{_{{\rm{SRD}} }}}(\gamma)$, the outage probability in \eqref{ap:P_out_sd_srd_0} can be expressed as
\begin{eqnarray}
P_{\rm out}^{_{(\rm SD+SRD)}} &\!\!=\!\!& \int\nolimits_0^{{\gamma _{{\rm{th}}}}} F_{_{|h_3|^2}} (\gamma_{\rm th} - \gamma)   f_{\gamma _{_{{\rm{SRD}} }}}(\gamma) {\rm{d}}\gamma \nonumber \\
&=&  P_{\rm out}^{_{(\rm SRD)}} - \int\nolimits_0^{{\gamma _{{\rm{th}}}}} \bar F_{_{|h_3|^2}} (\gamma_{\rm th} - \gamma) f_{\gamma _{_{{\rm{SRD}} }} }(\gamma) {\rm{d}}\gamma  \nonumber \\
& \mathop  = \limits^{(a)}  & P_{\rm out}^{_{(\rm SRD)}} \left(1-  \bar F_{_{|h_3|^2}} (\gamma_{\rm th} - \gamma_c) \right) \nonumber \\
&<& P_{\rm out}^{_{(\rm SRD)}} F_{_{|h_3|^2}} (\gamma_{\rm th}) \nonumber \\
&=& P_{\rm out}^{_{(\rm SD)}} P_{\rm out}^{_{(\rm SRD)}}, \label{ap:P_out_sd_srd}
\end{eqnarray}
where $P_{\rm out}^{_{(\rm SRD)}}=\int\nolimits_0^{\gamma_{\rm th}} f_{\gamma _{_{{\rm{SRD}} }}}(\gamma) {\rm{d}}\gamma$ is the outage probability achieved by the SWIET FDR and the constant $\gamma _{\rm c} \in [0, \gamma_{\rm th}]$. In the step (a) of \eqref{ap:P_out_sd_srd}, we have applied the weighted mean value theorem \cite[Theorem 3.16]{Calculus_vol1}.

\balance

\bibliography{IEEEabrv,IEEE_bib}


\end{document}